\documentclass[11pt,letterpaper]{article}

\usepackage[noend]{algorithmic}

\usepackage[margin=1in]{geometry}
\usepackage{lmodern}
\usepackage[T1]{fontenc}
\usepackage[utf8]{inputenc}
\usepackage[tbtags]{amsmath}
\allowdisplaybreaks
\usepackage{amssymb,amsthm,xspace}
\usepackage{mathtools}
\usepackage[hidelinks]{hyperref}
\usepackage[dvipsnames]{xcolor}

\usepackage[capitalise,nameinlink]{cleveref}
\crefname{ineq}{inequality}{inequalities}
\creflabelformat{ineq}{#2{\upshape(#1)}#3} 
\crefname{claim}{Claim}{Claims}

\usepackage[round]{natbib}
\usepackage{doi}
\usepackage{tikz} 
\usepackage{pgf,pgfplots}
\pgfplotsset{compat=1.14}
\usepgfplotslibrary{fillbetween}
\usepackage{hyphenat}
\usepackage[font=footnotesize]{caption}
\usepackage{subcaption}
\usepackage{appendix}
\usepackage{bm,booktabs,nicefrac,colortbl,setspace}
\usepackage[ruled,vlined,linesnumbered]{algorithm2e}

\SetKwFor{WP}{With probability}{}{:}{}
\SetKwBlock{Prob}{With probability}{end}
\makeatletter
\newcommand*{\algotitle}[2]{%
	\stepcounter{algocf}%
	\hypertarget{algocf.title.\theHalgocf}{}%
	\NR@gettitle{#1}%
	\label{#2}%
	\addtocounter{algocf}{-1}%
}
\usepackage[symbol]{footmisc}

\newtheorem{theorem}{Theorem}
\newtheorem{lemma}{Lemma}

\newtheorem{corollary}{Corollary}
\newtheorem{example}{Example}

\theoremstyle{definition}
\newtheorem{definition}{Definition}
\newtheorem{claim}{Claim}

\newcommand{\R}{\mathbb{R}}

\DeclareMathOperator*{\argmax}{arg\,max}

\newcommand{\E}[1]{\mathbb E \left[ #1 \right]}

\newcommand{\p}[1]{\mathbb P \left( #1 \right)}

\newcommand{\ind}[1]{\mathbf{1}_{#1}}
\newcommand{\F}{\mathcal{F}}

\newcommand{\e}{\varepsilon}
\newcommand{\N}{\mathcal{N}}
\newcommand{\sampling}{\textsc{SampleSeq}\xspace}

\newcommand{\threshold}{\textsc{ThreshSeq}\xspace}
\newcommand{\knapsackgamma}{$\gamma$-\textsc{ParKnapsack}\xspace}
\newcommand{\knapsack}{\textsc{ParKnapsack}\xspace}

\newcommand{\unMax}{\textsc{SubmodMax}\xspace}
\newcommand{\calE}{\mathcal{E}}
\newcommand{\G}{\mathcal{G}}
\newcommand\doubleplus{+\kern-1.3ex+\kern0.8ex}

\newcommand{\thresholdgamma}{$\gamma$-\textsc{ThreshSeq}\xspace}
\newcommand{\knapsackmonotone}{\textsc{ParKnapsack-Monotone}\xspace}
\newcommand{\knapsackgammamonotone}{$\gamma$-\textsc{ParKnapsack-Monotone}\xspace}
\newcommand{\cardinal}{\textsc{ParCardinal}\xspace}
\newcommand{\cardinalgamma}{$\gamma$-\textsc{ParCardinal}\xspace}

\newcommand{\OPT}{\textsc{OPT}}
\newcommand{\ALG}{\textsc{ALG}}

\begin{document}

\title{Submodular Maximization subject to a Knapsack Constraint: Combinatorial \!Algorithms \!with \!Near-optimal \!Adaptive \!Complexity}

\author{
Georgios Amanatidis
\thanks{Dept. of Mathematical Sciences, University of Essex, UK, \& Archimedes Unit, Athena Research Center, Greece.}
\and
Federico Fusco
\thanks{Dept. of Computer, Control, and Management 
Eng. ``Antonio Ruberti'', Sapienza University of Rome, 
Italy.
}
\thanks{Corresponding Author:  \texttt{fuscof@diag.uniroma1.it}}
\and
Philip Lazos
\thanks{Input Output, UK.}
\and
Stefano Leonardi{$^\dag$}
\and
Alberto Marchetti Spaccamela{$^\dag$}
\and
Rebecca Reiffenh{\"a}user\thanks{Institute for Logic, Language and Computation, University of Amsterdam, The Netherlands.}}


\maketitle

\begin{abstract}
Submodular maximization is a classic algorithmic problem with multiple applications in data mining and machine learning; there, the growing need to deal with massive instances motivates the design of algorithms balancing the quality of the solution with applicability. For the latter, an important measure is the \emph{adaptive complexity}, which captures the number of sequential rounds of parallel computation needed by an algorithm to terminate. In this work we obtain the first \emph{constant factor} approximation algorithm for non-monotone submodular maximization subject to a knapsack constraint with \emph{near-optimal} $O(\log n)$ adaptive complexity. Low adaptivity by itself, however, is not enough: a crucial feature to account for is represented by the total number of function evaluations (or value queries). Our algorithm asks $\tilde{O}(n^2)$ value queries, but can be modified to run with only $\tilde{O}(n)$ instead, while retaining a low adaptive complexity of $O(\log^2n)$. Besides the above improvement in adaptivity, this is also the first \emph{combinatorial} approach with sublinear adaptive complexity for the problem and yields algorithms comparable to the state-of-the-art even for the special cases of cardinality constraints or monotone objectives.
\end{abstract}

\paragraph{Version v2.} This version addresses a gap in the probabilistic analysis of the approximation guarantees in the previous version of this work, pointed out in \citet{CuiH23}. We provide a simple fix via a standard sampling routine while maintaining the same approximation guarantees and complexity bounds.

\section{Introduction}
Submodular optimization is a very popular topic relevant to various research areas as it captures the natural notion of {\em diminishing returns}. Its numerous applications include viral marketing \citep{HartlineMS08,KempeKT15}, data summarization \citep{TschiatschekIWB14,MirzasoleimanBK16,DuettingFLNZ22}, feature selection 
\citep{DasK08,DasK18,MirzasoleimanBL20} and clustering \citep{MirzasoleimanKSK13}. Notable examples from combinatorial optimization are cut functions in graphs and coverage functions \citep[see, e.g., Ch.~44 of][]{Schrijver03}.

Submodularity is often implicitly associated with monotonicity, and many results rely on that assumption \citep[see, e.g., the survey][]{BuchbinderF18}.
However, non-monotone submodular functions do naturally arise in applications, either directly or from combining monotone submodular objectives with modular {\em penalization} or {\em regularization} terms
\citep{HartlineMS08,TschiatschekIWB14,MirzasoleimanBK16,BreuerBS20,AmanatidisFLLR22}.
Additional constraints, like cardinality, matroid, knapsack, covering, and packing constraints, are prevalent in applications and have been extensively studied \citep[again, see][and references therein]{BuchbinderF18}. In this list, \emph{knapsack} constraints are among the most natural, as they capture limitations on  budget, time, or size of the elements. Like matroid constraints, they generalize cardinality constraints, yet they are not captured by the former.

The main computational bottleneck in submodular optimization comes from the necessity to repeatedly evaluate the objective function for various candidate sets. These \emph{value queries} are often notoriously heavy to compute, e.g., for exemplar-based clustering \citep{DueckF07}, log-determinant of submatrices \citep{KazemiZK18}, and accuracy of ML models \citep{DasK08,KhannaEDNG17}. 
With real-world instances of these problems growing to enormous sizes, simply reducing the number of queries is not always sufficient and parallelization has become an increasingly central paradigm \citep{BalkanskiBS18}. However, classic results in the area, often based on the greedy method, are inherently sequential: the intuitive approach of building a solution element-by-element contradicts the requirement of running \emph{independent} computations on many machines in parallel. 
The degree to which an algorithm can be parallelized is measured by the notion of \emph{adaptive complexity}, or adaptivity, introduced in \citet{BalkanskiBS18}. It is defined as the number of sequential rounds of parallel computation needed by the algorithm to terminate. In each of these rounds, polynomially many value queries may be issued, but they can only depend on the answers to queries relative to past rounds.

\paragraph{Contribution.}
 We propose the first combinatorial randomized algorithms for maximizing a (possibly) non-monotone submodular function subject to a knapsack constraint that combine constant approximation, low adaptive complexity, and a small number of queries. In particular, we obtain
\begin{itemize}
    \item a $9.465$-approximation algorithm, \knapsack (Algorithm \ref{alg:non_monotone_alg}), that has $O(\log n)$ adaptivity and uses $O(n^2\log^2 n)$ value queries, where $n$ is the cardinality of the ground set of the submodular function. This is the \emph{first} constant factor approximation algorithm for the problem with optimal adaptive complexity 
    up to a $O(\log \log n)$ factor (\Cref{thm:non_monotone}).
    \item a variant of our algorithm 
     with the same approximation, near-optimal $O(n\log^3 n)$ query complexity, and $O(\log^2 n)$ adaptivity (\Cref{cor:non_monotone-binary}). This is the first constant factor approximation algorithm that uses only $\tilde{O}(n)$ queries and has \emph{sublinear} adaptivity.
    \item $3$-approximation algorithms for \emph{monotone} objectives that combine $O(\log n)$ adaptivity with $O(n^2\log^2 n)$ total queries, and  $O(\log^2 n)$ adaptivity with $O(n\log^3 n)$ queries, respectively (\Cref{thm:monotone}). Even in the monotone setting, the latter is the first $O(1)$-approximation algorithm combining $\tilde{O}(n)$ queries and sublinear adaptivity. 
    \item $5.83$-approximation algorithms for \emph{cardinality} constraints (where the goal is to select the best subset of $k$ elements) that match or surpass the state-of-the-art when it comes to the combination of approximation, adaptivity and total queries (\Cref{thm:cardinality}).
    \item a method for interpolating between adaptivity and query complexity which applies to all these scenarios covered by  \Cref{thm:non_monotone,thm:monotone,thm:cardinality}, implying $O(\nicefrac{\log^2 n}{\log \gamma})$ adaptivity and $\tilde O(n \gamma )$ total queries, for any $\gamma>0$ (\Cref{thm:non_monotone-binary,cor:monotone,cor:cardinality}).
\end{itemize}
See \Cref{tab:our_results} for an overview of our results and a comparison with the state of the art.

\paragraph{Technical Challenges.} Like many existing works on adaptivity in submodular maximization \citep[e.g.,][]{BalkanskiS18,BalkanskiRS19}, we achieve poly-logarithmic adaptive complexity by sequentially adding to our candidate solution large subset of ``high value'' elements, while maintaining feasibility. In particular, following the approaches of \citet{BalkanskiRS19matroid} and \citet{BreuerBS20} for cardinality and matroid constraints, we compute these large (feasible) sets of high value elements by iteratively sampling feasible sequences of  elements and choosing suitable prefixes of them (in a single round of parallel computation). Crucially, knapsack constraints do not allow for the elegant counting arguments used in the case of cardinality or matroid constraints. The reason is that while the latter can be interpreted as a $1$-system, a knapsack constraint may induce a $\Theta(n)$-system\footnote{See \Cref{app:combinatorial} for a proof of this fact, as well as for formal definitions of matroids and $k$-systems.}, leading to poor results when naively adjusting existing approaches. Stated differently, for a knapsack constraint there is no clear indication of how long a sequence of feasible elements may be, whereas, e.g., for matroid constraints, it is always possible to complete a feasible set of cardinality $j$ to a basis adding exactly $k-j$ elements, where $k$ is the rank of the matroid.

A natural and very successful way of circumventing the resulting difficulties is to turn towards a \emph{continuous} version of the problem \citep[see, e.g.,][]{EneNV19,ChekuriQ19}. This, however, requires evaluating the objective function also for \emph{fractional} sets, i.e., such algorithms require access to an oracle for the multilinear relaxation and its gradient. Typically, these values are estimated by sampling, requiring $\tilde{\Theta}(n^2)$ samples \citep[see, e.g., Lemma 2.2 and Section 4.4 of][]{ChekuriQ19}. 
Our choice to avoid the resulting increase in query complexity and deal directly with the discreteness of the problem calls for specifically tailored algorithmic approaches.
Most crucially, our main subroutine (\threshold in \Cref{sec:non-monotone}) balances a suitable definition of \emph{good quality} candidates with a way to also reduce the size (not simply by cardinality, but a combination of overall cost and absolute marginal values) of the candidate set by a constant factor in each adaptive round.

Another main challenge is non-monotonicity. In presence of elements with negative marginals, the task of finding large feasible subsets of elements to add to our current solution becomes way more complicated: adding one single element with negative marginal contribution can arbitrarily deteriorate intuitive quality measures like the overall marginal density of the candidate set, causing a new adaptive round. 
Our approach combines carefully designed quality checks (the value and cost conditions in \threshold), with a separate handling of some elements with crucial properties, namely the elements with cost less than  $\nicefrac 1n$ of the budget and elements of maximum value.

\begin{table*}[t]
\renewcommand{\arraystretch}{1.2}
\setlength{\tabcolsep}{5pt}
\begin{center}
\begin{tabular}{llllll}
\hline
Reference                   & Objective & Constraint    & Approx.   & Adaptivity             & Queries          \\ \hline 
\citet{EneNV19}            & General   & 
Knapsack      & $\bm{e + \e}$   & 
$O(\log^2 n)$  \rule{0pt}{11pt}                  & 
$\tilde{O}(n^2)$          \\
\rowcolor{yellow!25!white} \Cref{thm:non_monotone} (this work)& General   & Knapsack      & $9.465 + \e$     & $\bm{O(\log n)}$  & $\tilde{O}(n^2)$\\
\rowcolor{yellow!25!white} \Cref{cor:non_monotone-binary} (this work)& General   & Knapsack      & $9.465 + \e$     & $O(\log^2 n)$  & $\bm{\tilde{O}(n)}$\\ \hline
\citet{EneNV19}      & Monotone  & Knapsack    & $\nicefrac{\bm e}{\bm {(e-1)}}\bm{+\e}$  & $\bm{O(\log n)}$    \rule{0pt}{11pt}            & $\tilde{O}(n^2)$          \\
\citet{ChekuriQ19}   & Monotone  & Knapsack    & $\nicefrac{\bm e}{\bm {(e-1)}}\bm{+\e}$  & $\bm{O(\log n)}$                & $\tilde{O}(n^2)$                 \\
\Cref{thm:monotone} (this work)         & Monotone  & Knapsack    & $3+\e$                   & $\bm{O(\log n)}$  &$\tilde{O}(n^2)$ \\
\Cref{cor:monotone} (this work)         & Monotone  & Knapsack    & $3+\e$                   & $O(\log^2 n)$  &$\bm{\tilde{O}(n)}$ \\\hline
\citet{EneN20}       & General   & Cardinality & $\bm{e + \e}$        & $\bm{O(\log n)}$   \rule{0pt}{11pt}             & $\tilde{O}(nk^2)$                  \\
\citet{Kuhnle21}     & General   & Cardinality & $6+\e$   & $\bm{O(\log n)}$  & $\bm{\tilde{O}(n)}$ \\
\citet{Kuhnle21}     & General   & Cardinality & $5.18 + \e$   & $O(\log^2 n)$  & $\bm{\tilde{O}(n)}$ \\
\Cref{thm:cardinality} (this work)   & General   & Cardinality & $5.83 + \e$          & $\bm{O(\log n)}$  & $\tilde{O}(nk)$ \\  \Cref{cor:cardinality} (this work)   & General   & Cardinality & $5.83 + \e$          & $O(\log n\log k)$  & $\bm{\tilde{O}(n)}$  \\ \hline
\end{tabular}
\caption{{ Our results---main result highlighted---compared to the state-of-the-art for low adaptivity. 
Bold indicates the best result(s) in each setting. In the last five rows $k$ stands for the size of the cardinality constraint. In the last two columns the dependence on $\e$ is omitted, while $\tilde{O}$ hides terms poly-logarithmic in $\e^{-1}$, $n$ and $k$. Note that our general adaptivity and query complexity trade-off is having $O\big(\frac{\log n}{\log \gamma}\big)$ adaptive complexity with $\tilde O(\gamma n)$ queries.}}
\label{tab:our_results}
\end{center}
\end{table*}

\paragraph{Related Work.}
Submodular maximization has been studied extensively since the seminal work of \citet{NemhauserWF78}. 
For \textit{monotone} submodular functions subject to a knapsack constraint 
the  $\nicefrac{e}{(e-1)}$-approximation algorithm of
\citet{Sviridenko04} is the best possible, unless 
$\text{P}=\text{NP}$ \citep{Feige98}.   
For the \emph{non-monotone} case,
a number of continuous greedy approaches \citep{FeldmanNS11,KulikST13,ChekuriVZ14} led to the current best factor of $e$ when a knapsack, or any kind of downward closed constraint, is involved. Combinatorial approaches \citep{GuptaRST10,AmanatidisFLLR22} achieve somewhat worse approximation, but are often significantly faster and thus relevant in practice. 

Adaptive complexity for submodular maximization was introduced by \citet{BalkanskiS18} for monotone objectives and a cardinality constraint, where they achieved an $O(1)$ approximation algorithm with $O(\log n)$ adaptivity, along with an almost matching lower bound: to get an $o(\log n)$  approximation, adaptivity must be $\Omega\big(\nicefrac{\log n}{\log \log n}\big)$. This result has been then improved \citep{BalkanskiRS19, EneN19,FahrbachMZ19,BreuerBS20} and recently \citet{ChenDK21} achieved an optimal $\nicefrac{e}{(e-1)}$-approximation in $O(\log n)$ adaptive rounds and (expected) $O(n)$ query complexity. 

The study of adaptivity for non-monotone objectives was initiated by \citet{BalkanskiBS18} 
again for a cardinality 
constraint, showing a constant approximation in $O(\log^2 n)$ adaptive rounds, later improved by \cite{FahrbachMZ19nonmonotone}, \cite{EneN20} and \cite{Kuhnle21}. Non-monotone maximization is also interesting in the unconstrained scenario. 
\citet{EneNV18} and \citet{FK19} achieved a $2+\e$ approximation  with constant adaptivity depending only on $\e$. Note that the algorithm of \citet{FK19} needs only $\tilde O(n)$ value queries, where the $\tilde{O}$ hides terms poly-logarithmic in $\e^{-1}$ and $n$.

More general constraints,  e.g., matroids and multiple packing constraints, have also been studied \citep{BalkanskiRS19matroid,EneNV19,ChekuriQ19matroid,ChekuriQ19}.  \citet{EneNV19} and \citet{ChekuriQ19} provide low adaptivity results---$O(\log^2 n)$ for non-monotone and $O(\log n)$ for monotone---via continuous approaches for packing constraints (and, thus, for knapsack constraints as well). We refer to \Cref{tab:our_results} for an overview; note that the query complexity of the algorithms is stated with respect to queries to $f$ and not to its multilinear extension. \citet{ChekuriQ19} also provide two combinatorial algorithms for the monotone case: one with optimal approximation and adaptivity but $O(n^4)$ value queries, and one with linear query complexity, optimal adaptivity but an approximation factor parameterized by the cost of the most expensive element, which can be arbitrarily bad. We mention that algorithms with low adaptive complexity have been used as subroutines for more complex problems, e.g., adaptive submodular maximization \citep{EsfandiariKM21} and fully dynamic submodular maximization \citep{LattanziMNTZ20}. 

\paragraph{Conference version and follow up work.} A conference version of this paper appeared in the Thirty-eighth International Conference on Machine Learning (ICML 21)
\citep{AmanatidisFLLMR21}. In this version we address a gap in the probabilistic analysis of the approximation
guarantees of \knapsack which has been pointed out in \citet{Cui000LZ23} (but only detailed in \citet{CuiH23}). \citet{Cui000LZ23} also provide a fix and offer an $8$-approximation algorithm. Their approach is substantially similar to ours: e.g., their routine GetSEQ is equivalent to our \sampling, their routine RandBatch is equivalent to our \threshold; the improved approximation factor is achieved by a more careful handling of the ``small cost'' elements.

\section{Preliminaries}\label{sec:prelims}

    Let $f:2^{\N} \to \R$ be a set function  over a ground set $\N$ of $n$ elements. We say that $f$ is {\em non-negative} if $f(S)\ge 0$, for all $S \subseteq \N$, and \emph{monotone} if $f(S) \le f(T)$, for all $S,T \subseteq \N$ such that $S \subseteq T.$ For any $S, T \subseteq \N$, $f(S\,|\,T)$ denotes the {\em marginal value} of $S$ with respect to $T$, i.e., $f(S \cup T) - f(T)$. To ease readability, we write $f(x\,|\,T)$ instead of $f(\{x\}\,|\,T)$. The function $f$ is \emph{submodular} if $f(x\,|\,T) \le f(x\,|\,S)$, $\forall S, T \subseteq \N$ with $S\subseteq T$ and $x \notin  T$.

    \paragraph{Submodular maximization with knapsack constraint.} The main focus of this work is the maximization of non-negative (possibly non-monotone) submodular functions subject to a \emph{knapsack constraint}. Formally, we are given a budget $B > 0$, a non-negative submodular function $f$ and a non-negative additive cost function $c:2^{\N} \to \R_{\ge0}$. The goal is to find a set with cost at most $B$ that maximizes the function $f$:
    \[
        O^*\in \argmax_{T \subseteq\N\,: \,c(T) \le B} f(T).
    \] 
    Let $\OPT = f(O^*)$ denote the value of such an optimal set. Given a (randomized) algorithm for the problem, let $\ALG$ denote the (expected) value of its output. We say that the algorithm is a $\beta$-approximation algorithm, for $\beta \ge 1$, if the inequality $\OPT \le \beta \cdot\ALG$ holds for any instance. Throughout this work, we assume, without loss of generality, that $\max_{x \in \N}c(x)\le B$ (again, for the sake of readability, we write $c(x)$ for the cost of a singleton $\{x\}$). 

    We often refer to cardinality constraints in our work. As is typically the case in the submodular maximization literature, this means that the feasible sets have cardinality at most a given number $k$, rather than a cost which is at most $B$. Note that a cardinality constraint is a special cases of a knapsack constraint: it suffices to consider the cost function $c(T) = |T|$ for all $T \subseteq \N$, and take the budget $B$ to be the upper bound $k$ on the cardinality.

    \paragraph{Query and adaptive complexity.} We assume access to $f$ through value queries, i.e., for each $S \subseteq \N$, an oracle returns $f(S)$ in constant time. We define here the two crucial notions of computational complexity of a submodular maximization algorithm: adaptive and query complexity.    
    \begin{definition}[Adaptive complexity]
        Given a value oracle for $f$, the \emph{adaptive complexity} or \emph{adaptivity} of an algorithm is the minimum number of rounds in which the algorithm makes $O(\textrm{poly}(n))$ \emph{independent} queries to the evaluation oracle. In each adaptive round the queries may {\em only} depend on the answers to queries from past rounds.
    \end{definition} 
    At a high level, the adaptive complexity is an abstract measure of how sequential/parallelizable an algorithm is. As an example, the standard greedy algorithm of \citet{NemhauserW78} that repeatedly adds the element with largest marginal value to the current solution is inherently sequential (the value queries issued to identify the next element to add to the current solution crucially depend on the answers to the previously issued value queries) and , thus, has linear adaptive complexity. Conversely, an algorithm that builds the solution using only the singleton values has adaptive complexity $1$ (since the singleton values can be computed in a single round of independent computation).

    \begin{definition}[Query complexity]
    Given a value oracle for $f$, the \emph{query complexity} of an algorithm is the total number of value queries it issues.
    \end{definition}
    
    \paragraph{Useful properties of submodular functions.} We state some widely known properties of submodular functions that are extensively used in the rest of the paper. The first Lemma summarizes two equivalent definitions of submodular functions shown by \citet{NemhauserWF78}. 
    \begin{lemma}
    \label{lem:folklore}
        Let $f:2^{\N} \to \R$ be a submodular function and $S,T,U$ be any subsets of $\N$, with $S \subseteq T$. Then the following two properties hold:
        \begin{itemize}
            \item[$i)$] $f(U\,|\,T) \le f(U\,|\,S)$
            \item[$ii)$] $f(S\,|\,T) \le \sum_{x \in S} f(x\,|\,T)$.
        \end{itemize}
    \end{lemma}
    The second Lemma (originally proved in \citet{FeigeMV11} and that we report here as \citet[Lemma 2.2 of][]{BuchbinderFNS14}) is an important tool for tackling non-monotonicity. 
    \begin{lemma}[Sampling Lemma]
    \label{lem:sampling}
        Let $f:2^{\N} \to \R$ be a submodular function and $X$ be any random subset of $\N$, where each element of $\N$ belongs with $X$ with probability at most $p$. Then the following inequality holds:
        \[
            \E{f(X)} \ge (1-p) f(\emptyset).
        \]
    \end{lemma}
    In the spirit of this Lemma, for any set $X$, and probability $p \in [0,1]$ we denote with $X(p)$ a random set generated as follows: each $x$ in $X$ belongs to $X(p)$ with probability $p$, independently from all the others.

    \paragraph{Uncostrained submodular maximization oracle.}
    Finally, in this paper we assume  access to \unMax, an unconstrained submodular maximization oracle. 
    For instance, this can be implemented via the combinatorial algorithm 
    of \citet{FK19}, which
    outputs a $(2+\e)$-approximation of $\max_{T\subseteq \N} f(T)$
    for any given precision $\e$ in $O(\frac 1 \e \log \frac 1 \e)$ adaptive rounds and linear query complexity.
    For our experiments, we use the much simpler $4$-approximation of \citet{FeigeMV11}, which terminates after one single adaptive round.

    \paragraph{The role of the precision parameter $\e$.} 
    We discuss the role of the precision parameter $\e \in (0,1)$ in our results and the formal meaning of the $O(\e)$ notation in the Theorem statements. Formally, with an approximation guarantee of $\alpha + O(\e)$ we mean that for any constant $\delta \in (0,1)$ there exists a constant $C_{\delta}$ such that running the algorithm with parameter $\e \in (0,\delta)$ yields an $(\alpha + C_{\delta} \e)$-approximation, for any choice of $\e.$ Alternatively, running the algorithm with parameter $\nicefrac{\e}{C_{\delta}}$ yields an $(\alpha + \e)$-approximation.  At a high level, $\e$  is intended as an arbitrarily small but constant parameter that can be tuned following the needs of the algorithm designer. Accordingly, the terms poly($\nicefrac 1{\e})$ and poly-log($\nicefrac 1{\e}$) in the adaptive and query complexity bounds are to be considered as constants that multiply to the leading terms (which depend on $n$). This is consistent with what is typically done in the submodular maximization literature (e.g., the precision parameter of Lazy-Greedy \citep{Minoux78} and Sample-Greedy \citep{AmanatidisFLLR22}).

\section{Non-Monotone Submodular Maximization}
\label{sec:non-monotone}

    In this Section we present and analyze our algorithm for non-monotone submodular maximization with knapsack constraint: \knapsack. Given the complexity of the algorithm, we proceed in three steps. First, in \Cref{subsec:sampleseq}, we present a simple routine, \sampling, that we use to generate sequences of elements that fit into the budget. Then, in \Cref{subsec:thresholdseq} we introduce our main subroutine, \threshold, which uses the sequences generated by \sampling to iteratively construct a solution in few adaptive rounds. Finally, we present and analyze the main algorithm, \knapsack, in \Cref{subsec:knapsack}.

    \subsection{{\normalfont{\textsc{SampleGreedy}}}: Our sampling procedure}
    \label{subsec:sampleseq}

        Our first routine is \sampling, which simply samples a sequence of elements that respect the knapsack constraint. Formally, given two disjoint sets $S$ and $X$ (with $S \cap X = \emptyset$), and a budget $B$, \sampling outputs a sequence $A$ each element of which is sequentially drawn uniformly at random among the remaining elements of $X$ that do not cause $S\cup A$ to exceed the budget. We refer to the pseudocode for the details. 

        \begin{algorithm}
        \caption{\sampling$(S,X,B)$}
        \label{alg:sampling_sequence}
        \begin{spacing}{1.15}
        \begin{algorithmic}[1]
            \STATE \textbf{Environment:} submodular function $f$ and additive cost function $c$
            \STATE \textbf{Input:} disjoint sets $S$ and $X$, budget $B>0$
            \STATE $A \leftarrow [ \ \ ]$
            \STATE $i\gets 1$
            \STATE $X \gets \{ x \in X: c(x) + c(S) \le B\}$
            \WHILE {$X \neq \emptyset$}
            \STATE Draw $a_i$ uniformly at random from $X$ \label{line:sampling}
            \STATE $A \gets [a_1,\dots,a_{i-1}, a_i]$
            \STATE $X \gets \{ x \in X\setminus \{a_i\}: c(x) + c(A) + c(S) \le B\}$
            \STATE $i \gets i+1$
            \ENDWHILE
            \STATE \textbf{return} $A = [a_1,a_2,\dots,a_d]$
        \end{algorithmic}
        \end{spacing}
        \end{algorithm}

        We observe that the length of the output sequence $A$ may vary dramatically due to the random elements sampled in line \ref{line:sampling} and the fact that knapsack constraints are $O(n)$-systems (see 
        \Cref{app:combinatorial}). Moreover, by design, \sampling only adds elements to $A$ such that the invariant $c(S) + c(A) \le B$ is maintained for any realization of the randomness in line \ref{line:sampling}. Thus, we directly have the following Lemma.

        \begin{lemma}
        \label{lem:budget_sampling}
            For any input $S,X$ and $B$, with $c(S) \le B$, the sequence $A$ outputted by \sampling is such that $c(A) + c(S) \le B$, with probability $1$.
        \end{lemma}

    \subsection{{\normalfont{\textsc{ThreshSeq}}}: Finding the right prefix}
    \label{subsec:thresholdseq}

        The crux of our approach lies in the routine \threshold. It receives as input a set $X\subseteq \N$ of elements, a threshold $\tau$, a budget $B$, and a parameter $\e \in (0,1)$, and outputs a set $S\subseteq X$ with $c(S) \le B$ such that, in expectation, each element in $S$ has marginal density that is at least approximately $\tau$ (\Cref{lem:non_monotone_val}). We present \threshold and then analyze its complexity (\Cref{lem:non_monotone_adaptivity}) and stopping condition (\Cref{lem:non_monotone}). \threshold iteratively constructs a candidate solution $S$ by using \sampling. In each iteration of a \emph{while} loop (line \ref{line:while}) a sequence $A$ is sampled via \sampling, then for each prefix $A_i = \{a_1, \dots, a_i\}$ of $A$ the following three subsets of $X$ are computed: 
        \begin{itemize}
            \item The feasible elements $X_i$ that can be added to $A_i \cup S$ without violating the budget constraint (line \ref{line:X_i})
            \item The ``good elements'' $G_i$ that are contained in $X_i$ and exhibit marginal density with respect to $S \cup A_i$ larger than $\tau$ (line \ref{line:G_i})
            \item The ``negative elements'' $E_i$ that are contained in $X_i$ and have negative marginal value with respect to $S \cup A_i$ (line \ref{line:E_i})
        \end{itemize}
        We say that a prefix $A_i$ respects the cost condition if $c(G_i) > (1-\e)c(X)$ (line \ref{line:cost_condition}), while it respects the value condition (line \ref{line:value_condition}) if the following inequality holds:
        \[
            \sum_{x \in G_i} \e  f(x\,|\,S \cup A_i) > \sum_{x \in E_i}|f(x\,|\,S \cup A_i)|.
        \]
        \threshold adds to the current solution $S$ the shortest prefix $A_{k^*}$ that violates either of these two conditions (lines \ref{line:kstar} and \ref{line:Supdate}). Note that, by design, $i^*$ and $j^*$ (and thus $k^*$) are always well defined: $X_{d}$ (i.e., the set of feasible elements at the end of the sequence $A = [a_1,\dots, a_d]$) is empty, this implies that also $G_d$ and $E_d$ are empty as well. In any iteration of the while loop we either have that $i^* \le j^*$, in which case we say that the cost condition has been triggered, or $j^* < i^*$, in which case we say that the value condition has been triggered. \threshold exits the while loop when either set $X$ is empty, or the value condition has been triggered $\ell$ times.

        \begin{algorithm}[t]       \caption{\threshold$(X,\tau,\e,B)$}
        \label{alg:nonmonotone_sequence}
        \begin{spacing}{1.15}
        \begin{algorithmic}[1]
            \STATE \textbf{Environment:} Submodular function $f$ and additive cost function $c$
            \STATE \textbf{Input:} set $X$ of elements, threshold $\tau>0$, precision $\e\in(0,1)$, and budget $B$
            \STATE  \textbf{Initialization: }$ S \gets \emptyset $, $\textrm{ctr} \gets 0$
            \STATE $X \gets \{x \in X: f(x) \ge \tau \, c(x)\}$ \COMMENT{Filtering low-density elements in $X$}
            \WHILE{$X \neq \emptyset$ and $\textrm{ctr}<\nicefrac{1}{\e^2}$} \label{line:while}
            \STATE  $[a_1,a_2,\dots ,a_d]\leftarrow  \sampling(S,X,B)$\;
            \FOR{$i=1,\dots,d$}
            \STATE  $A_i\gets \{a_1,a_2,\dots,a_i\}$ \COMMENT{Candidate sequence}
            \STATE  \label{line:X_i}
              $X_i \gets \{a \in X\setminus A_i: c(a) + c(S \cup A_i) \le B\}$ \COMMENT{Feasible elements}
            \STATE  \label{line:G_i}
              $G_i \gets \{a \in X_i: f(a\,|\,S \cup A_i) \ge \tau \cdot c(a)\}$ \COMMENT{Good elements}
            \STATE   \label{line:E_i}
              $E_i \gets \{a \in X_i: f(a\,|\,S \cup A_i) < 0\}$
              \COMMENT{Negative elements}
            \ENDFOR
            \STATE    \label{line:cost_condition}
            $i^* \gets \min \{i: c(G_i) \le (1-\e)c(X)\}$ \COMMENT{Cost condition}
            \STATE    \label{line:value_condition}
                $\displaystyle j^* \gets \min \Big\{j:  \sum_{x \in G_j} \e  f(x\,|\,S \cup A_j) \le \sum_{x \in E_j}|f(x\,|\,S \cup A_j)|\Big\}$
            \COMMENT{Value condition}
            \STATE  $k^*\gets\min\{i^*,j^*\}$ \label{line:kstar}
            \STATE $S \leftarrow S \cup A_{k^*}$, $X \gets G_{k^*}$ \label{line:Supdate}
            \STATE \textbf{if} $j^* < i^*$ \textbf{then} $\textrm{ctr} \leftarrow \textrm{ctr}+1$ \COMMENT{Counter of the value condition}
            \ENDWHILE
            \STATE \textbf{return} {$S$}
        \end{algorithmic}
        \end{spacing}
        \end{algorithm}

        Consider now the complexity of \threshold. The crucial observation is that each iteration of the while loop can be carried over in {\em one single} adaptive round of calculation: once the prefix $A$ is drawn, all the auxiliary sets and the conditions can be independently computed for each prefix. All the value queries depend only on $S$ and $A$, and not on the answer to the queries issued to compute $k^*$. The cost conditions and the value conditions are designed to bound the number of times the while loop is repeated: in each adaptive round, either the counter \textrm{ctr} of the value conditions increases by $1$ (this can happen at most $\nicefrac 1{\e^2}$ times) or the overall cost of the candidate set decreases by a factor of $(1-\e)$. To have a clear bound on how many times the cost condition can be triggered we introduce the quantity $\kappa_X$ that measures how much the cost function varies in a set $X$: $ \kappa_X = \max_{x,y \in X} \nicefrac{c(x)}{c(y)}$; we adopt the convention that the previous ratio is set to infinity if there exists an element with zero cost in $X$ (and $\kappa_\emptyset = 0$). We formalize these observations in the following Lemma.
        \begin{lemma}
        \label{lem:non_monotone_adaptivity}
            For any input $X \subseteq \N$, $\tau >0$, $\e \in (0,1)$ and $B>0$, \threshold\ terminates in $O\left( \frac{1}{\e}\log (n\kappa_X)\right)$ adaptive rounds and issues $O\left( \frac{n^2}{\e}\log (n\kappa_X)\right)$ value queries.
        \end{lemma}
        \begin{proof}
            The adaptive rounds correspond to iterations of the while loop. As already argued, once a new sequence is drawn by \sampling, all the value queries needed are deterministically induced by it and hence can be assigned to different independent machines. Gathering this information we can determine $k^*$ and start another iteration of the while loop. Bounding the number of such iterations where the value condition is triggered  is easy, since it is forced to be at most $\nicefrac{1}{\e^2}$. 
            
            For the cost condition we use the geometric decrease in the total cost of $X$: every time it is triggered (i.e., $i^* \le j^*$ and thus $k^*=i^*$ in line \ref{line:kstar}), the total cost of the feasible elements $X$ is decreased by at least a $(1-\e)$ factor.
            At the beginning of the algorithm, that cost is at most $Cn$, with $C=\max_{x \in X} c(x)$, and it needs to decrease below $c= \min_{x \in X} c(x)$ to ensure that $X = \emptyset$. Let $r$ the number of iterations of the while loop where the cost condition was triggered ($i^*\le k^*$). In the worst case we need $Cn(1-\e)^r < c$,  meaning that the total number of iterations of the while loop is upper bounded by $\nicefrac{\log \left(n \cdot \kappa_X \right)}{\e} + \nicefrac{1}{\e^2}$ (which is in $O(\nicefrac{\log \left(n \cdot \kappa_X \right)}{\e})$ due to the assumption of $\nicefrac 1\e\le\log n).$
            
            Finally, the query complexity is just a $n^2$ factor greater than the adaptivity: each adaptive round entails $O(n^2)$ value queries, since the length of the sequence outputted by \sampling \ may be linear in $n$ and for each prefix the value of the marginals of all the remaining elements has to be considered.
        \end{proof}
        Having settled the adaptive and query complexity of \threshold, we move to proving that our conditions ensure good expected marginal density. \threshold adds to $S$ prefixes $A_i = [a_1,\dots,a_i]$ such that for all $j<i$ the average contribution to $S \cup A_j$ of the elements in $X \setminus A_j$ is comparable to $\tau$. Then the expected value of $f(A_i\,|\,S)$ should be comparable to $\tau \, \E{c(A_i)}$. To compute $A_i$ in one single parallel round, one can {\em a posteriori} compute for each prefix $A_j$ of $A$ the {\em a priori } (with respect to the uniform samples) expected marginal value of $a_{j+1}$; with $a_{j+1}$ drawn uniformly at random from the elements in $X \setminus A_j$ still fitting the budget, this means simply averaging over their marginal densities. We formalize this argument in the following Lemma. 
        \begin{lemma}
            \label{lem:non_monotone_val}
            For any input $X \subseteq \N$,  $\tau>0$,  $\e \in (0,1)$,  and  $B>0$, the random set $S$ outputted by \threshold is such that $c(S) \le B$ and the following inequality holds true:
            \[
                {\E{f(S)}\ge \allowbreak (1-\e)^2\tau \,\E{c(S)}}.
            \]
        \end{lemma}
        \begin{proof}
            We first note that $c(S) \le B$ with probability $1$ by \Cref{lem:budget_sampling} and the fact that the costs are non-negative. We focus then on the rest of the statement. \threshold adds a chunk of elements to the solution in each iteration of the \emph{while} loop. This, along with the fact that each of these chunks is an ordered prefix of a sequence outputted by \sampling, induces a total ordering on the elements in $S$. To facilitate the presentation of this proof, we assume, without loss of generality, that the elements of $S$ are added one after the other, according to this total order. 
            Let us call the $t$-th such element $s_t$, and let $\F_t$ denote the random story of the algorithm up to, but excluding, the adding of $s_{t}$ to its chunk's random sequence. We show that whenever any $s_t$ is added, its expected marginal density is at least $(1-\e)^2\tau$.
        
            Fix any $s_t$ and any story $\F_t$ of the algorithm up to that point. and consider the iteration of the while loop in which it is added to the solution. We denote with $ S_{old}$ the partial solution at the beginning of that while loop, with $A$ the sequence drawn in that iteration by \sampling and with $X$ the candidate set at that point: 
            \[
                X = \{x \in \N: f(x\,|\,S_{old}) \ge \tau \cdot c(x), c(x) + c(S_{old}) \le B\}.
            \]
            Let $A_{(t)}$ be the prefix of $A$ up to, and excluding, $s_t$. Then $S_t = S_{old} \cup A_{(t)}$ is the set of all elements added to the solution before $s_t$.  Note that, given $\F_t$, the sets $X$, $S_{old}$ and $A_{(t)}$ are deterministic, while the rest of $A$ is random. Recall that $s_t$ is drawn uniformly at random from $X_{(t)} = \{x\in X \setminus A_{(t)}|c(S_t)+c(x)\leq B\}$. 
            We need to show that 
            \begin{equation}
                \label{eq:value2}
                \E{f(s_t\,|\,S_t)\,|\,\F_t} \ge  (1-\e)^2\tau \,\E{c(s_t)\,|\,\F_t}\,,
            \end{equation}where the randomness is  with respect to the uniform sampling in $X_{(t)}$. 
            
            If $s_t$ is the first element in $A$, then \Cref{eq:value2} holds because all the elements in $X$ exhibit a marginal density greater than $\tau$. If $s_t$ is not the first element, it means that the value and cost condition were not triggered for the previous element. Let $G$ and $E$ the sets of the good and negative elements with respect to $S_t$, i.e., $G = \{x \in X_{(t)}: f(x\,|\,S_t) \ge \tau c(x)\}$ and $E = \{x \in X_{(t)}: f(x\,|\,S_t)< 0\}$, which are also deterministically defined by $\F_t$. Finally, let $p_x$ be $\p{s_{t}=x\,|\,\F_{t}}$ which is equal to  $|X_{(t)}|^{-1}$ for all $x \in X_{(t)}$ and zero otherwise, then
            \begin{align}
                \nonumber
                \mathbb{E}[f(s_t\,|\,S_{t})\,|\,\F_{t}] &- (1-\varepsilon)^2 \tau \,\E{c(s_t)\,|\,\F_{t}}\\
                \nonumber
                =&  \sum_{x \in X} p_x  f(x\,|\,S_{t}) - (1-\varepsilon)^2 \tau \sum_{x \in X} p_x c(x)   \\\nonumber
                \ge&  \sum_{x \in G\cup E}p_x f(x\,|\,S_{t}) 
                - (1-\varepsilon)^2 \tau \sum_{x \in X} p_x c(x)\\
                \label{eq:values_bound}
                =& \, \e \sum_{x \in G}p_x f(x\,|\,S_{t}) - \sum_{x \in E}p_x|f(x\,|\,S_{t})| \\
                \label{eq:costs_bound}
                & +(1-\e)\tau \Big[ \sum_{x \in G}
                p_xc(x) - (1-\varepsilon)\sum_{x \in X}p_xc(x)\Big{]}\\
                \label{eq:threshold_bound}
                & + (1-\e)\sum_{x \in G}p_x \Big[f(x\,|\,S_{t})- \tau c(x)\Big]  \ge 0 \,.
            \end{align}
            Expressions \eqref{eq:values_bound} and \eqref{eq:costs_bound} are nonnegative since the value and cost conditions were not triggered before adding $s_t$. Expression \eqref{eq:threshold_bound} is nonnegative by the definition of $G$. 
            
            We have then shown that, for all $t$ and $\mathcal{F}_t$, the expected marginal density of the $t$-th element (if any) added by our algorithm is large enough. Next we carefully apply conditional expectations to get the desired bound. We have already argued how the algorithm induces an ordering on the elements it adds to the solution, so that they can be pictured as being added one after the other. To avoid complication in the analysis we suppose that after the algorithm stops it keeps on adding dummy elements of no cost and no value, so that in total it runs for $n$ time steps. 
            Consider the filtration $\{\F_t\}_{t=1}^n$ generated by the stochastic process associated to the algorithm, where $\F_t$ narrates what happens up to the point when element $s_t$ is considered. So that $\F_1$ is empty and $\F_n$ contains all the story of the algorithm except for the last---possibly dummy---added element.
            From the above analysis we know that for each time $t \in \{1,\dots,n\}$ and any possible story $\F_t$ of the algorithm we obtain \Cref{eq:value2}. Note that the inequality holds also if one considers the dummy elements after the actual termination of the algorithm. 
            \begin{align*}
            \nonumber
                \E{f(S)} &= \E{\sum_{t=1}^n f(s_t\,|\,S_t)}
                = \sum_{t=1}^n \E{f(s_t\,|\,S_t)}
                = \sum_{t=1}^n \E{\E{f(s_t\,|\,S_t)\,|\,\F_t}}\\
                &= \E{\sum_{t=1}^n \E{f(s_t\,|\,S_t)\,|\,\F_t}}
                \ge  (1-\e)^2\tau \,\E{\sum_{t=1}^n \E{c(s_t)\,|\,\F_t}}\\
                &=  (1-\e)^2\tau \,\E{\sum_{t=1}^n c(s_t)}
                =  (1-\e)^2\tau \,\E{c(S)}\,.
            \end{align*}
            The second and fourth equalities hold by linearity of expectation, the third and fifth equalities hold by the law of total expectation. 
            Finally, the inequality follows from monotonicity of the conditional expectation and inequality \eqref{eq:value2}.
        \end{proof}

        The previous Lemma establishes that $S$ has expected density comparable to our threshold $\tau$. We move on to showing that when \threshold\ terminates, we can bound the value of good candidates still fitting inside the budget that are left outside the solution.

        \begin{lemma}
        \label{lem:non_monotone}
        For any input $X \subseteq \N$,  $\tau>0$,  $\e \in (0,1)$,  and  $B>0$, the random set $S$ outputted by \threshold respects that $\sum_{x \in G}f(x\,|\,S) \le \e f(S)$, where $G$ contains the good elements that still fit in the remaining budget: 
        \[
            G = \{x \in X\setminus S:\allowbreak f(x\,|\,S)  \ge  \tau c(x), \ c(x) + c(S)\le B \}.
        \]
        \end{lemma}
        \begin{proof}
            \threshold terminates in one of two cases. Either $X$ is empty, meaning that there are no elements still fitting in the budget whose marginal density is greater than $\tau$---and in that case the inequality we want to prove trivially holds---or the value condition has been triggered $\ell = \lceil \tfrac 1{\e^2} \rceil$ times.  
            
            For the latter, suppose that the value condition was triggered for the $i$th time during iteration $t_i$ of the while loop. Denote by $ S_{t_i}$ the solution at the end of that iteration. We are interested in the sets $X_{j^*}, G_{j^*}, E_{j^*}$ of that particular iteration of the while loop. In order to be consistent across iterations, we use $X_{(i)}$, $G_{(i)}$, and $E_{(i)}$ to denote these sets for iteration $t_i$. 
            Since the value condition was triggered during $t_i$, we have $\e\sum_{x \in G_{(i)}}f(x\,|\,S_{t_i}) \le \sum_{x \in E_{(i)}} |f(x\,|\,S_{t_i})|$.
            Clearly, $G_{(\ell)}$ is what we denoted by $G$ in the statement and $S_{t_{\ell}}$ is $S$. Also notice that $E_{(j)}\cap E_{(k)} = \emptyset$ for $j\neq k$.
            
            Now, by non-negativity of $f$ and \Cref{lem:folklore}, we have
            \[
                0 \le f\Big( S_{t_{\ell}} \cup {\textstyle \bigcup\limits_{i=1}^{\ell}} E_{(i)} \Big) \le f(S_{t_{\ell}}) + \sum_{i=1}^{\ell}\sum_{x \in E_{(i)}} f(x\,|\,S_{t_i})\,.
            \]
            Rearranging the terms and using the value condition, we get
            \begin{align*}
                f(S_{t_{\ell}}) &\ge \sum_{i=1}^{\ell}\sum_{x \in E_{(i)}} |f(x\,|\,S_{t_i})| \ge \sum_{i=1}^{\ell}\e \sum_{x \in G_{(i)}} f(x\,|\,S_{t_i})
                \ge \frac 1 \e \sum_{x \in G_{(\ell)}}f(x\,|\,S_{t_{\ell}})\,.
            \end{align*}
            The last inequality follows from submodularity  and the fact that that $G_{(1)}\supseteq \ldots\supseteq G_{(\ell)}$.
            \end{proof}

    \subsection{{\normalfont{\textsc{ParKnapsack}}}: Our algorithm}
    \label{subsec:knapsack}
    
        We are ready to present our algorithm \knapsack for non-monotone submodular maximization subject to a knapsack constraint. We refer to the pseudocode for the details. 
        \knapsack partitions the ground set $\N$ into two sets, according to the cost function $c$: the ``small'' elements $\N_-$, each with cost smaller than $\nicefrac Bn$, and the ``large'' elements $\N_+=\N \setminus \N_-$ (line \ref{line:partition}). The set $\N_-$ is fed to the low adaptive complexity unconstrained maximization routine \unMax\ (line \ref{line:unconstrained}) to obtain a candidate solution $S_-$ as discussed in \Cref{sec:prelims}. 
        For the large elements we need more work. For $\Theta(\nicefrac{\log(n)}{\e})$ parallel guesses $\tau_i$ of the ``right'' threshold (lines \ref{line:for_guesses} and \ref{line:guesses}) two independent blocks of instructions run:
        \begin{itemize}
            \item[(i)] $\Theta(\nicefrac{1}{\e}\log({\nicefrac 1{\e}}))$ fresh sub-samples $H^{i,j} \sim \N_+(p)$\footnote{Recall, $X \sim \N_+(p)$ means that $X$ is a random set to which each element of $\N_+$ belong independently with probability $p$.} are drawn from $\N_+$ (line \ref{line:sample_j}) and \threshold is called on them, with threshold $\tau_i$ (line \ref{line:thresh_j}). The best set computed in this first block (over all the thresholds $\tau_i$ and repetitions $j$) is called $S_{\ge}$ (line \ref{line:candidates}).
            \item[(ii)] A set $H_i\sim \N_+(p)$ is sampled once (line \ref{line:sample_l}) and  \threshold is called $\Theta(\nicefrac{1}{\e}\log({\nicefrac 1{\e}}))$ independent times on it (line \ref{line:for_l} and \ref{line:thresh_l}). The best set computed in this second block (over all the thresholds and independent calls to \threshold) is called $S_{<}$ (line \ref{line:candidates}).
        \end{itemize}
        Finally, \knapsack returns the best between $S_-$, $S_{\ge}$, $S_<$, and $x^*$ (line \ref{line:argmax}), which is the singleton with the largest value (line \ref{line:max}). Note, the partition between $\N_+$ and $\N_-$ is critical in bounding the adaptivity of \threshold, as $\kappa_{\N_+}\le n$ (and the same holds for any subset of $\mathcal N_+$).
        
        \begin{algorithm}
        \caption{\knapsack}
        \label{alg:non_monotone_alg}
        \begin{spacing}{1.1}
        \begin{algorithmic}[1]
        \STATE \textbf{Environment:} Submodular function $f$, additive cost function $c$ and budget $B$
        \STATE \textbf{Input:} Precision parameter $\e \in (0,1)$ 
        \STATE $\alpha \gets \nicefrac{(3-\sqrt{3})}6$, $p \gets \nicefrac{(\sqrt 3 - 1)}2$ 
        \label{line:init}
        \STATE $\N_- \gets \{x \in \N: c(x) < \nicefrac{B}{n}\}$, $\N_+ \leftarrow \N\setminus \mathcal{N_-}$ \label{line:partition}
        \STATE $x^* \gets \argmax_{x \in \N}f(x)$, $\hat \tau \gets \alpha n \cdot \nicefrac {f(x^*)}{B}$ \label{line:max}
        \STATE $S_- \gets \unMax (\N_-,\e)$\label{line:unconstrained} 
        \FOR{$i=0, 1,\dots,\lceil{\nicefrac {\log n}
{\e}\rceil}$ in parallel } \label{line:for_guesses}
        \STATE $\tau_i \gets \hat \tau \cdot (1-\e)^i$ \label{line:guesses}
        \FOR{$j=1, 2,\dots, \lceil\nicefrac 1\e \log(\nicefrac{1}{\e})\rceil$ in parallel} \label{line:parallel}
        \STATE Sample independently a subset $H^{i,j} \gets \N_+(p)$
        \label{line:sample_j}
        \STATE $S^{i,j}\gets \threshold(H^{i,j},\tau_i,\e,B)$ \label{line:thresh_j}
        \ENDFOR
        \STATE Sample independently a subset $H_i \gets \N_+(p)$ \label{line:sample_l}
        \FOR{$\ell = 1, 2, \dots, \lceil\nicefrac 1\e \log(\nicefrac{1}{\e})\rceil$ in parallel} \label{line:for_l}
            \STATE $S_{i,\ell}\gets \threshold(H_i,\tau_i,\e,B)$ \label{line:thresh_l}
        \ENDFOR
        \ENDFOR
        \STATE $S_{\ge} \in \argmax_{i,j} \{f(S^{i,j})\}$, $S_{<} \in \argmax_{i,\ell} \{f(S_{i,\ell})\}$ \label{line:candidates}
        \STATE \textbf{return}  $T \in\argmax\{f(S_{\ge}),f(S_<),f(x^*),f(S_-)\}$ \label{line:argmax}
        \end{algorithmic}
        \end{spacing}
        \end{algorithm}

        At a high level, the main argument behind the approximation guarantee result is simple. Assume to have access to the value of $\OPT$, then running \threshold with threshold proportional to $\OPT$ gives rise to two cases: either our low-adaptivity routine \threshold fills a constant factor of the budget (let's say half of it) in expectation --- in which case \Cref{lem:budget_sampling} ensures us the good quality of the solution--- (this is captured by $S_{\ge}$), or the routine stops earlier --- meaning that there are no ``good'' elements left to add --- as stated in \Cref{lem:non_monotone} (this is captured by $S_<$). In practice, given the probabilistic nature of our arguments, we need to carefully design a sampling strategy so that overall the final solution yields a good approximation. This motivates the two blocks of introductions above, where (i) corresponds to the first case, i.e., a good fraction of the budget is filled in expectation, and (ii) to the second. 
        \begin{theorem}
        \label{thm:non_monotone}
            For any $\e \in (0,1)$, \knapsack with precision parameter $\e$ is a $(9.465 + \Theta(\e))$-approximation algorithm that runs in $O(\frac{1}{\e}\log n)$ adaptive rounds and issues $O(\frac{n^2}{\e^3}\log^2 \!n\log\!{\frac 1\e})$ value queries.
        \end{theorem}
        \begin{proof}
            We first argue about the query and adaptive complexity. The algorithm communicates with the value oracle in four lines of the algorithm that can be run in parallel as they do not need each other's output. First, in line \ref{line:max} it computes $x^*$ in one adaptive round and $O(n)$ value queries. Second, in line \ref{line:unconstrained} the routine \unMax is called; we use the combinatorial $(2+\e$)-approximation by \citet{ChenDK21} which runs in $O(\nicefrac 1{\e} \log (\nicefrac 1{\e}))$ adaptive rounds and exhibits linear query complexity. 
            Finally, in lines \ref{line:thresh_j} and \ref{line:thresh_l} the routine \threshold is called on some subsets of $\N_+$. By definition of small elements, it holds that $\kappa_{\N_+} \le n$, thus, by \Cref{lem:non_monotone_adaptivity}, each call of \threshold has $O(\nicefrac {\log n}{\e} )$ adaptive complexity and issues $O(\nicefrac {n^2}{\e} \log n)$ value queries. 
            Lines \ref{line:thresh_j} and \ref{line:thresh_l} of \knapsack are repeated independently $O(\nicefrac{1}{\e^2} \log n \log (\nicefrac 1{\e}))$ times; this does not affect the adaptive complexity but only the query complexity in a multiplicative way. All in all, the total query and adaptive complexity of the algorithm is dominated by this latter term (the repetitions of lines \ref{line:thresh_j} and \ref{line:thresh_l}). This concludes the proof of the adaptive and query complexity of \knapsack. 

            We move our attention to the approximation guarantee and we start introducing some notation. Let $O^*$ be any optimal solution with $f(O^*) = \OPT$, and $O^+,O^-$ be its intersections with $\N_+$ and $\N_-$ respectively, and define set $O$ as an optimal solution in $\N_+$, i.e., 
            \[
                O\in \argmax\{f(T) : T\subseteq \N_+,\ c(T)\le B\}.
            \]
            Note, all the sets mentioned in the previous sentence are deterministic. Set $\tau^*=\alpha \cdot \nicefrac{\OPT}B$ (where parameter $\alpha$ is defined in line \ref{line:init});  by the parallel guesses of the threshold, there exists one iteration of the outer for loop (line \ref{line:for_guesses}) with threshold $\tau_{i^*}$ that is close enough to $\tau^*$, i.e., such that 
            \begin{equation}
            \label{eq:tau_istar}
                (1-\e)\tau^* \le \tau_{i^*} < \tau^*.
            \end{equation}
            The existence of such $\tau$ follows from  the definitions of $\tau^*$ and $\hat \tau$ (see line \ref{line:max}) and the fact that $n \cdot f(x^*) \ge f(O^*) \ge f(x^*)$ (the first inequality follows by submodularity and the fact that there are $n$ elements in $\N$, and the second one by optimality of $O^*$). In the analysis we only need the random sets corresponding to this threshold $\tau_{i^*}$, for this reason we often omit the index $i^*$. 
            
            Given any set $X \subseteq \N_+$, we say that it ``fills'' the budget if the expected cost of the output of $\threshold(X,\tau,\e, B)$ is at least $(1-\e)B/2.$ Note, the expectation here is with respect to the randomness of \threshold. The crux of our case analysis is $\mathcal E_{\ge}$, which is defined as the event that a subset $H \sim \mathcal N_+(p)$ fills the budget. We denote as $\mathcal E_<$ the complementary event. We have now two cases, depending on the value of $\p{\mathcal E_{\ge}}$:
            \begin{itemize}
                \item If $\p{\mathcal E_\ge} \ge \e$, then we prove that at least one of the samples $H^j$ (the index $i^*$ is omitted) drawn in the first block fills the budget, and therefore its corresponding solution (and thus also $S_{\ge}$) is a good approximation to $\OPT$ (\Cref{cl:bound_ge}).
                \item If $\p{\mathcal E_\ge}<\e$ then our analysis contemplates multiple intermediate steps  (\Cref{cl:unconstrained,cl:p_bound,cl:anyreal}) and culminates in \Cref{cl:bound_le}, where the value of $O$ is upper bounded using $f(S_-)$ and $f(S_<)$.
            \end{itemize}
            Note: while $\mathcal E_\ge$ is a random event (its randomness only resides in the random experiment of sampling $H$), the value of its probability is a fixed property of the input (i.e., the set $\mathcal N$ and the submodular function $f$).

            We start by the first case: if the probability of filling the budget is at least $\e$, then $S_{\ge}$ is an $\tfrac \alpha 2$ approximation of $\OPT$ (the value of $\alpha$ is defined in line \ref{line:init}).
            \begin{claim}
            \label{cl:bound_ge}
                If $\p{\mathcal E_\ge} \ge \e$, then the following inequality holds:
                \[
                    \OPT \le \frac{2}{\alpha(1-\e)^5} \E{f(S_\ge)} \le (2(3 + \sqrt{3}) + O(\e))\ALG.
                \]
            \end{claim}
            \begin{proof}[Proof of \Cref{cl:bound_ge}]
                We focus on the first block of instructions (lines \ref{line:sample_j} and \ref{line:thresh_j}) in the for loop corresponding to the right threshold $\tau_{i^*}$. Fix any realization of some $H^j$ in line \ref{line:sample_j} of \knapsack that fills the budget (we denote with $\mathcal E_{\ge}^{j}$ such event and $\mathcal E_{<}^{j}$ its complementary), it means that $\E{c(S^j) \mid \mathcal E_{\ge}^{j}} \ge (1-\e) B/2$ (where the expectation is with respect to the run of \threshold, not the sampling). By \Cref{lem:non_monotone_val} we have the following chain of inequalities:
                \begin{align*}
                    \E{f(S^j) \mid \mathcal E_{\ge}^{j}} &\ge (1-\e)^2\tau_{i^*} \,\E{c(S_L) \mid \mathcal E_{\ge}^{j}} \tag*{(by \Cref{lem:non_monotone_val})}\\
                    &\ge (1-\e)^3 \frac{\alpha}{B}\OPT\cdot \E{c(S_L)\mid \mathcal E_{\ge}^{j}}\tag*{\text{(\Cref{eq:tau_istar})}}\\
                    &\ge (1-\e)^4 \frac{\alpha}{2}\OPT\tag*{(\text{$H^j$ fills the budget})}.    
                \end{align*}
                Now, let $\mathcal G_{\ge}$ the union of the $\mathcal E_\ge^j$; i.e., $\mathcal G_{\ge}$ is the event that at least one of the $H^j$ fills the budget. Clearly, by the previous chain of inequalities it holds that under $\mathcal G_{\ge}$ the expected value of at least one of the $S^j$ (and therefore of $S_\ge$) is larger than $(1-\e)^4 \frac{\alpha}{2}\OPT$. To conclude, we only need to argue that $\G_\ge$ has a good probability of happening:
                \[
                    \p{\G_{\ge}^C} = \p{\bigcap_j \mathcal E^j_{<}} = \prod_j\p{\mathcal E^j_{<}} = \prod_j \p{\mathcal E_<} \ge \e,
                \]
                where the second equality is due to the independent parallel samples (repetitions of line \ref{line:sample_j} in the inner loop) and the inequality is due to the assumption on $\p{\mathcal E_{\ge}}$ and the number of these parallel samples. So we have that $\p{\mathcal G_{\ge}} \ge 1-\e$. We can finally apply the law of total expectation:
                \[
                    \E{f(S_{\ge})} \ge \E{f(S_{\ge}) \mid \mathcal G_{\ge}} \p{\mathcal G_{\ge}}\ge  (1-\e)^5 \frac{\alpha}2\OPT.
                \]
                If we plug in the chosen value of $\alpha = \nicefrac{(3 - \sqrt 3)}6$ we get the desired result:
                \[
                     \OPT \le \frac{2(3 + \sqrt{3})}{(1-\e)^5} \E{f(S_\ge)} \le (2(3 + \sqrt{3}) + O(\e))\ALG.
                \]
                From the above inequality it is clear that we use the $O(\e)$ notation to simplify the approximation-factor-dependence on $\e$ to the first order term. If for example we want to study any $\e \in (0,\nicefrac 15)$, then we have the uniform bound $\nicefrac{2(3 + \sqrt{3})}{(1-\e)^5} \le 2(3 + \sqrt{3}) + 100 \e$. 
            \end{proof}
            We move now to consider the other case, when the probability of filling the budget is low. We start relating $f(O^-)$ with $f(S_-)$.       
            \begin{claim}
            \label{cl:unconstrained}
                The following inequality holds: $f(O^-) \le (2+\e) \cdot f(S_-)$.
            \end{claim}
            \begin{proof}[Proof of \Cref{cl:unconstrained}]
                We can upper bound $f(O^-)$ with the unconstrained max on $\N_-$, since there are at most $n$ elements in $\N_-$ and each of them costs at most $\nicefrac{B}n$ (for a total cost smaller than the budget $B$). Using the combinatorial $(2+\e)$-approximation of \citet{FK19}, we get 
                \[
                    f(O^-) \le (2+\e) \cdot f(S_-). \qedhere
                \]
            \end{proof}
                    
            Now that we have bounded the fraction of $\OPT$ that depends on the small elements $\N_-$ we move our attention to $O^+$. By maximality of $O$ in $\N_+$, it holds that $f(O^+) \le f(O)$, so we upper bound the latter. We focus on the second block of instructions (lines \ref{line:sample_l} to \ref{line:thresh_l}), in the iteration of the outer loop corresponding to the right threshold (so that we may omit the dependence on the index $i$). Denote with $H$ the sampled set (in line \ref{line:sample_l}), with $O_H = O \cap H$ its intersections with $O$. We define the (random) $S^*$ as follows: if there exists a repetition $\ell$ such that $c(S_{\ell})<\nicefrac B2$, then $S^* = S_{\ell}$ (with ties broken arbitrarily), otherwise $S^* = \emptyset$. Set $S^*$ is crucial in the remaining analysis.
            
            \begin{claim}
            \label{cl:p_bound}
                The following relation
                holds: 
                \[
                    p (1-p) f(O) \le \E{f( O_H \cup S^*)}.
                \]
            \end{claim}
            \begin{proof}[Proof of \Cref{cl:p_bound}] Consider any ordering of the elements in $O=\{o_1,o_2, \dots, o_m\}$ for some $m$. We have the following:
            \begin{align}
                \E{f(O_H)} &= \E{\sum_{i=1}^m \ind{o_i \in H}f(o_i|\{o_1, \dots, o_{i-1}\} \cap H)}\tag*{(Telescopic argument)}\\
                &\ge \E{\sum_{i=1}^m \ind{o_i \in H}f(o_i|\{o_1, \dots, o_{i-1}\})}\tag*{(By submodularity)}
                \\
                &= \sum_{i=1}^m \p{o_i \in H}f(o_i|\{o_1, \dots, o_{i-1}\})\tag*{(By linearity)}\\
                \label[ineq]{eq:O_HvsO}
                &= p \cdot \sum_{i=1}^m f(o_i|\{o_1, \dots, o_{i-1}\}) = p {f(O)}. 
            \end{align}
            Fix now any possible realization of $O_H$ and apply \Cref{lem:sampling} on the submodular function $g(\,\cdot\,) = f(\,\cdot\, \cup O_H)$ and the random set $S^*\setminus O_H$. Indeed, elements belong to $H$ with probability $p$, thus any element in $\N\setminus O_H$ belongs to $S^* \subseteq H$ with probability at most $p$. There is a crucial observation here: the random set $S^*$ depends on the particular realization of $H$, but every fixed element $x$ belongs to $H$ with probability exactly $p$, so overall it belongs to $S^*$ with probability at most $p$ (where the ``at most'' is due to the randomness of \threshold in the realized $H$). We obtain
            \[\E{f(S^* \cup O_H) \,|\, O_H} \ge (1-p) f(O_H) \, ,\]
            which can be combined with \Cref{eq:O_HvsO} (for the first inequality) and the law of total expectation (for the equality) to conclude the proof:
            \[
                p (1-p) f(O) \le (1-p) \E{f(O_H)} \le \E{\E{f(S^* \cup O_H)\,|\,O_H}} = \E{f(S^* \cup O_H)}. \qedhere
            \]
            \end{proof}
            In the previous Claim we related the expected value of $f(O_H \cup S^*)$ with $f(O)$; in a further step we relate $f(O_H \cup S^*)$ with $f(S_<)$. We focus once again on the iteration of the outer for loop corresponding to the right threshold $\tau_{i^*}$, so we omit the dependence on the index $i$. For any repetition of line \ref{line:thresh_l}, we call $\mathcal G_{\ell}$ the event that $c(S_{\ell}) < B/2$ and $\mathcal G_<$ their union. 
        \begin{claim}
        \label{cl:anyreal}
            Fix any run of the algorithm such that $\mathcal G_<$ is realized, the following relation holds:
            \[
                f(S^* \cup O_H) \le (1+\e) f(S_<) +\tau^* {c(O_H)}.
            \]
        \end{claim}
        \begin{proof}[Proof of \Cref{cl:anyreal}]
            We condition with respect to $\mathcal G_<$, so we focus on any $\ell$ such that $\G_{\ell}$ is realized. We split the high marginal density elements of $H$ (with respect to $S_{\ell}$) into two sets, depending on whether these are feasible or not:
            \begin{align*}
                G & =\{x \in H: f(x\,|\,S_{\ell}) \ge \tau c(x), c(x) + c(S_{\ell}) \le B\}\,,\\
                \tilde G & =\{x \in H: f(x\,|\,S_{\ell}) \ge \tau c(x), \ c(x) + c(S_{\ell}) > B\}\,.
            \end{align*}
            We decompose the contribution of the elements in $O_H \setminus S_{\ell}$ (via \Cref{lem:folklore}) using $G$ and $\tilde G$. Note that since we are conditioning on $c(S_{\ell})<B/2$, there can be at most one single element $\tilde x$ in $\tilde G \cap O_H$. Indeed, each such element has cost at least $B-c(S_{\ell})>B/2$, yet belongs to the set $O$ that has cost at most $B$. 
            
            Before delving into the calculations, we observe that is safe to assume that $f(x^*) < \alpha \cdot \OPT/2$: if this was not the case, in fact, we could directly claim that $f(x^*)$ (and therefore $\ALG$) yields the approximation guarantee stated in the Theorem. Having settled this simple corner case, we distinguish two cases depending on whether $\tilde G \cap O_H$ is empty or not. 
            
            If $\tilde G \cap O_H \neq \emptyset$, then $\tilde G \cap O_H =\{\tilde x\}$, i.e., it is a singleton as we already argued. In that case,
            \begin{align*}
            \nonumber
                f(S_{\ell}\cup O_H) &\le f(S_{\ell}) +\sum_{x \in \tilde G \cap O_H}f(x\,|\,S_{\ell})  +\sum_{x \in G \cap O_H}f(x\,|\,S_{\ell})+ \!\!\sum_{x \in O_H \setminus (G \cup \tilde G)} \!\!f(x\,|\,S_{\ell}) \\
                &\le (1+\e) f(S_{\ell}) +f(\tilde x\,|\,S_{\ell})+ \!\!\sum_{x \in O_H \setminus (G \cup \tilde G)} \!\!f(x\,|\,S_{\ell}) \tag*{\text{(\Cref{lem:non_monotone} and definition of $\tilde x$)}}\\
                &< (1+\e) f(S_{\ell}) + f(\tilde x) + \tau^* ({c(O_H)} - c(\tilde x)) \tag*{\text{(Submodularity; definitions of $G$, $\tilde G$, $\tau$)}}\\
                &< (1+\e) f(S_{\ell}) + \tfrac{\alpha}{2} \OPT + \tau^* ({c(O_H)} - \tfrac B2) \tag*{\text{($f(\tilde x)\le f(x^*) < \tfrac{\alpha}{2} \OPT$; $c(\tilde x) > \frac{B}{2}$)}}\\
                &< (1+\e) f(S_<) +  \tfrac{\alpha}{2} \OPT+ \tau^* {c(O_H)} - \tfrac{\alpha}{2} \OPT  \tag*{\text{($f(S_<) \ge f(S_{\ell})$; definition of $\tau^*$)}}\\
                &= (1+\e) f(S_<) + \tau^* {c(O_H)}  \,.
            \end{align*}

            On the other hand, if $\tilde G \cap O_H = \emptyset$, the chain of inequalities is similar but simpler.  We have
        \begin{align*}
        \nonumber
            f(S_{\ell} \cup O_H) &\le f(S_{\ell}) +\sum_{x \in G \cap O_H}f(x\,|\,S_{\ell})+ \sum_{x \in O_H \setminus G} f(x\,|\,S_{\ell}) \nonumber \\
            &\le (1+\e) f(S_{\ell}) + \!\!\sum_{x \in O_H \setminus (G \cup \tilde G)} \!\!f(x\,|\,S_{\ell}) \nonumber \\
            &< (1+\e) f(S_<) + \tau^* {c(O_H)} \,. 
        \end{align*}
        All in all, we have proved the desired inequality.
        \end{proof}
        
        To conclude the argument relative to $\p{\mathcal E_\ge} < \e$, we are left with relating the value of $S_<$ with that of $S^*$.

        \begin{claim}
        \label{cl:bound_le}
            If $\p{\mathcal E_\ge} < \e$, then the following inequality holds:
            \[
                \OPT \le (2(3 + \sqrt{3}) + O(\e)) \ALG. 
            \]
        \end{claim}
        \begin{proof}[Proof of \Cref{cl:bound_le}]
        Recall, $\mathcal E_<$ is the event that the sampled $H$ does not fill the budget, and we are assuming that has probability at least $(1-\e)$.
        For any iteration $\ell$ of the inner loop of the second block, the probability of $\mathcal G_{\ell}$ (i.e., the event that $c(S_{\ell})$ has cost less than $B/2$) conditioning on $\mathcal E_<$ is at least $\varepsilon$. To see this:
        \[
            (1-\e)\frac{B}{2} > \E{c(S_{\ell}) \mid \mathcal E_<} 
            \ge \E{c(S_{\ell})\,|\,\mathcal E_<\cap (\mathcal G_{\ell})^C}\p{(\mathcal G_{\ell})^C\mid \calE_<} \ge \p{(\mathcal G^{\ell})^C\mid \calE_<} \frac{B}{2}\,.
        \]
        So, the probability that $S^* = S_{\ell}$ for some $\ell$ under $\mathcal E_<$ is at least $(1-\e)$. Combining this inequality with the assumption on $\p{\mathcal E_<}$, we have that $\mathcal E_<$ and $\G_<$ are both realized with large probability:
        \begin{equation}
        \label{eq:EandG}
            \p{\mathcal E_< \cap \G_<} = \p{\G_< \mid \mathcal E_<} \p{\mathcal E_< } \ge (1-\e)^2.
        \end{equation}
        The first application of this inequality is to combine it with \Cref{cl:p_bound}: 
        \begin{align}
            p(1-p) f(O) &\le \E{f(S^* \cup O_H)} \tag{\Cref{cl:p_bound}}\\
        \nonumber
                &\le \E{f(S^* \cup O_H) \mid \mathcal E_<\cap\mathcal G_<} + 2 f(O) \p{\left(\mathcal E_{<} \cap\G_<\right)^C}\tag{\Cref{eq:EandG}} \\
                \label{eq:step1}
                &\le \E{f(S^* \cup O_H) \mid \mathcal E_<\cap\mathcal G_<} + 4 \e f(O).
        \end{align} 
        At this point, let's apply \Cref{cl:anyreal}: 
        \begin{align}
            \E{f(S^*\cup O_H) \mid \mathcal E_<\cap \mathcal G} &\le (1+\e) \E{f(S_<) \mid \mathcal E_<\cap\mathcal G} + \tau^* \E{c(O_H)\mid \mathcal E_< \cap \mathcal G} \tag{\Cref{cl:anyreal}}\\
            \label{eq:step2}
            &\le (1+\e) \E{f(S_<) \mid \mathcal E_< \cap\mathcal G} + \alpha \frac{p}{(1-\e)^2}\OPT,
        \end{align}
        where in the last inequality we used the definition of $\tau^*$ and that:
        \[
        pB = \E{c(O_H)} \ge \E{c(O_H) \mid \G_<\cap\mathcal E_<}\p{\G_< \cap \mathcal E_<} \ge \E{c(O_H) \mid \G\cap\mathcal E_<} (1-\e)^2.  
        \]
        Finally, \Cref{eq:EandG} can be used to prove that the expected value of $f(S_<)$ is not influenced much by the conditioning on $\mathcal E_< \cap \G_<$:
        \begin{align}
        \label{eq:step3}    \E{f(S_<)} 
            &\ge \E{f(S_<) \mid \G \cap \mathcal E_<} (1-\e)^2,
        \end{align}
        because $f$ is non-negativ.
        
        We can combine the three inequalities (first \Cref{eq:step1}, then \Cref{eq:step2} and finally \Cref{eq:step3}) together and obtain the following:
        \begin{align}
        \nonumber
            p(1-p)f(O) &\le \E{f(S^* \cup O_H) \mid \mathcal E_< \cap\mathcal G} + 4 \e f(O)\\
        \nonumber
            &\le (1+\e) \E{f(S_<) \mid \mathcal E_< \cap \mathcal G} + \alpha \frac{p}{(1-\e)^2}\OPT + 4 \e f(O)\\
            \label{eq:very_last}
            &\le \frac{1+\e }{(1-\e)^2}\E{f(S_<)} + \alpha \frac{p}{(1-\e)^2}\OPT + 4 \e \OPT.
        \end{align}
       This inequality, together with \Cref{cl:unconstrained}, yields:
        \begin{align*}
            p(1-p)\OPT &\le p(1-p)f(O^-) + p(1-p)f(O)\\
            &\le p(1-p) (2+\e)f(S_-) + \frac{1+\e }{(1-\e)^2}\E{f(S_<)} + \alpha \frac{p}{(1-\e)^2}\OPT + 4 \e \OPT\\
            &\le \left[ p(1-p)(2+\e) + \frac{1+\e }{(1-\e)^2}\right] \ALG + \left(\alpha \frac{p}{(1-\e)^2} + 4 \e\right) \OPT.
        \end{align*}
        Finally, we can rearrange everything to obtain:
        \[
            \OPT \le \frac{ p(1-p)(2+\e)(1-\e )^2 + (1+\e)}{p(1-p)(1-\e)^2 - \alpha p - 4 \e(1-\e)^2}\ALG
        \]
        If we plug in the chosen value of $\alpha$ and $p$, i.e., $\alpha = \nicefrac{(3 - \sqrt 3)}6$ and $p = \nicefrac{(\sqrt 3 - 1)}2$, we get the desired inequality:
        \[
            \OPT \le (2(3 + \sqrt{3}) + O(\e)) \ALG.\qedhere
        \]
\end{proof}
    Combining \Cref{cl:bound_ge}, \Cref{cl:bound_le} and the observation on the case when $f(x^*)\ge \alpha /2$ (as in the proof of \Cref{cl:anyreal}) concludes the proof of the Theorem.
\end{proof}


\section{Variants and Implications}
\label{sec:variants}
    
    An exciting feature of our combinatorial approach is that---with few modifications---it yields a number of algorithms that match or improve the state-of-the-art in other scenarios, namely non-monotone submodular maximization with cardinality constraint (in \Cref{sec:non-monotone}) and monotone submodular maximization with knapsack constraints (in \Cref{sec:monotone}). Moreover, \threshold, as well as all its variants, can be tweaked to exhibit linear query complexity, at the cost of an extra $O(\log n)$ term in the adaptivity. More generally, it is possible to continuously interpolate between these two regimes, namely high adaptivity \& low queries - low adaptivity \& many queries, in an explicit way (\Cref{sec:trade_off}).

\subsection{Trading off Adaptivity and Query Complexity}
\label{sec:trade_off}

    We begin with an observation on the possible trade-offs between adaptivity and query complexity. \threshold needs $\tilde O(n^2)$ value queries because, in each iteration of its while loop (line \ref{line:while}) it locates $k^*$ by computing $O(n)$ marginal values with respect to the $O(n)$ prefixes of the sampled sequence $A$. Each iteration of the while loop is then characterized by a single adaptive round of computation (all the value queries are determined by $A$ and $S$ that are known immediately) and $\tilde O (n^2)$ value queries. If the mechanism designer is willing to trade some adaptivity with query complexity it is possible to spare $\Theta(\nicefrac{n}{\log n})$ value queries at the cost of $O(\log n)$ extra adaptive rounds, by using \emph{binary search} to locate $k^*$ in the while loop of \threshold. With this approach, the number of prefixes $A_i$ for which the algorithm needs to compute the marginals with respect to decreases to $O(\log n)$, at the cost of introducing some dependency in the value queries issued: for instance, the choice of the second prefix considered depends on the outcome of the value calls relative to the first prefix. For this approach to work, there is a last challenge to overcome: in order to find $k^*$ via binary search,  a carefully modified version of the value condition is needed, since the one used in \threshold exhibits a multi-modal behaviour. In particular, we want to tweak the value condition so that if it is triggered for a certain prefix $A_i$\footnote{Recall, prefix $A_i$ triggers the value condition if $i$ is considered in the $\argmax$ in line \ref{line:value_condition}: in formula, if it holds that $\e \sum_{x \in G_i}f(x\mid S \cup A_i) \le \sum_{x \in E_j} |f(x\mid S \cup A_j)|$)} it remains activated for all  $A_j$ for $j\ge i$ in the specific while loop iteration. 
    
    Following this approach, we can actually obtain a more fine-grained result, making explicit the trade off between adaptivity and query complexity by using $\gamma$-ary search instead of binary search. $\gamma$-ary search generalizes binary search and consists in dividing recursively the array in $\gamma$ equally spaced intervals and evaluate the conditions on the extremes of these intervals to locate $k^*.$ It is easy to observe that $\gamma$-ary search on an array of length $n$ terminates in $O(\nicefrac{\log n}{\log \gamma})$ steps while evaluating the conditions $O(\gamma)$ times at every step.

    \paragraph{The cost condition is already unimodal.} Fix any iteration of the while loop and call $S$ the solution at the beginning of it, $A$ the sequence drawn by \sampling (line \ref{line:sampling} of \threshold), $\{X_i\}_i$ the sequence of the sets of elements still fitting into the budget relative to each prefix $A_i$ (line \ref{line:X_i}) and $G_i$ (line \ref{line:G_i}), respectively $E_i$ (line \ref{line:E_i}), the subsets of $X_i$ containing the \emph{good}, i.e., marginal density greater than $\tau$,  and \emph{bad}, i.e., negative marginal density, elements, respectively. The cost condition (recall, a prefix $A_i$ respects the cost condition when $c(G_i) > (1-\e)c(X)$, see line \ref{line:cost_condition}) is clearly unimodal: the $\{X_i\}_i$ is a decreasing sequence of sets and hence $c(X_i)$ is a non-increasing sequence of costs, while $c(X)$ stays fixed: as soon as the cost of $X_i$ drops below $(1-\e)c(X)$ it stays there for all the prefixes longer than $A_i$.

    \paragraph{The value condition is not unimodal.} For the value condition we need a bit more work. The value condition in \threshold is verified for a prefix $A_i$ if the following inequality holds (as in line \ref{line:value_condition}):
    \[
        \e \,\sum_{x \in G_i} f(x\,|\,S \cup A_i) > \sum_{x \in E_i}|f(x\,|\,S \cup A_i)|.
    \]
    This condition is not unimodal: it may be the case that for a certain prefix $A_i$ the previous inequality is not verified (i.e., the value condition is triggered), but the inequality switches direction for some other prefix $A_j$ with $j > i$, in the same iteration of the while loop. This can happen for one of two reasons: either elements with negative marginals are added to the solution or they are thrown away due to budget constraint.  We want a modification which is robust to these corner cases. To this end, we add to the value condition the absolute contribution of two sets of items. We refer to the pseudocode of \thresholdgamma for the details.

    \begin{algorithm}[t]       \caption{\thresholdgamma$(X,\tau,\e,B)$}
        \label{alg:gamma_sequence}
        \begin{spacing}{1.15}
        \begin{algorithmic}[1]
            \STATE \textbf{Environment:} Submodular function $f$ and additive cost function $c$
            \STATE \textbf{Input:} set $X$ of elements, threshold $\tau>0$, precision $\e\in(0,1)$, budget $B$ and $\gamma \in \mathbb{N}$
            \STATE  \textbf{Initialization: }$ S \gets \emptyset $, $\textrm{ctr} \gets 0$
            \STATE $X \gets \{x \in X: f(x) \ge \tau \, c(x)\}$
            \WHILE{$X \neq \emptyset$ and $\textrm{ctr}<\nicefrac{1}{\e^2}$} 
            \STATE  $[a_1,a_2,\dots ,a_d]\leftarrow  \sampling(S,X,B)$\;
            \STATE \label{line:gammary} Find $k^*\gets\min\{i^*,j^*\}$ using $\gamma$-ary search \textbf{where:}
            \STATE \quad $A_i\gets \{a_1,a_2,\dots,a_i\}$ 
            \STATE \quad $X_i \gets \{a \in X\setminus A_i: c(a) + c(S \cup A_i) \le B\}$ 
            \STATE  \quad $G_i \gets \{a \in X_i: f(a\,|\,S \cup A_i) \ge \tau \cdot c(a)\}$
            \STATE  \label{line:old_negative}\quad $E_i \gets \{a \in X: f(a\,|\,S \cup A_i) < 0\}$ \COMMENT{Negative elements in $X$}
            \STATE \label{line:new_negative}\quad $\calE_i=\{a_t \in A_i: f(a_t\,|\,S \cup A_{t-1}) < 0\}$ \COMMENT{Negative elements in $A_i$}
            \STATE   \quad $i^* \gets \min \{i: c(G_i) \le (1-\e)c(X)\}$ 
            \STATE    \label{line:new_value_condition} \quad 
                $\displaystyle j^* \gets \min \Big\{j:  \e \sum_{x \in G_i}f(x\,|\,S \cup A_i) \le \sum_{x \in E_i} |f(x\,|\,S \cup A_i)| + \sum_{a_j \in \calE_i} |f(a_j\,|\,S \cup A_j)|\Big\}$
            \STATE $S \leftarrow S \cup A_{k^*}$, $X \gets G_{k^*}$
            \STATE \textbf{if} $j^* < i^*$ \textbf{then} $\textrm{ctr} \leftarrow \textrm{ctr}+1$ 
            \ENDWHILE
            \STATE \label{line:postprocessing} $\overline S \gets \{s_t \in S \mid f(s_t \mid \{s_1, s_2, \dots, s_{t-1}\}) > 0\}$
            \STATE \textbf{return} {$\overline S$}
        \end{algorithmic}
        \end{spacing}
        \end{algorithm}
    
    \paragraph{\thresholdgamma.} First, we redefine the set $E_i$ to contain also all the bad elements considered in that while loop, regardless of the budget condition, i.e., $E_i \gets \{a \in X: f(a\,|\,S \cup A_i) < 0\}$ (see line \ref{line:old_negative} of \thresholdgamma). Moreover for each prefix $A_i$, we define $\mathcal{E}_i$ as the set of all the items in the prefix $A_i$ which added negative marginal when inserted in the solution, i.e., $\calE_i=\{a_t \in A_i: f(a_t\,|\,S \cup A_{t-1}) < 0\}$ (see line \ref{line:new_negative}). The new value condition (see line \ref{line:new_value_condition}) then reads:
    \begin{equation}
        \label{eq:new_value_condition}
        \e \sum_{x \in G_i}f(x\,|\,S \cup A_i) > \sum_{x \in E_i} |f(x\,|\,S \cup A_i)| + \sum_{a_j \in \calE_i} |f(a_j\,|\,S \cup A_j)|\,.
    \end{equation}
    \begin{claim}
        The new value condition in line \ref{line:new_value_condition} of \thresholdgamma is unimodal.
    \end{claim}
    \begin{proof}
        We need to prove if the value condition is triggered for prefix $A_i$ (meaning that equation \ref{eq:new_value_condition} holds with $\le$), then it stays triggered for all prefixes $A_j$ with $j\ge i$. The left hand side of the condition is monotonically decreasing in $i$, while the right hand side is monotonically increasing, by submodularity and the fact that now $\{E_t \cup \calE_t\}_t$ is an increasing set sequence. 
    \end{proof}
    Now that both the cost and the value conditions are unimodal, it is clear that it is possible to perform $\gamma$-ary search in line \ref{line:gammary} of \thresholdgamma to find $k^*$. This means that the algorithm needs to issue way less value queries (as it needs to evaluate only logarithmically many prefixes and not all of them), at the cost of a worst adaptivity. While earlier it was possible to run in parallel the verification of the value and cost condition, now the algorithm needs to sequentially issue those queries, as the $\gamma$-ary search predicates. 
    Apart from this different value condition, in line \ref{line:postprocessing}, we add a post processing step before returning the new chunk of elements. Let $S = \{s_1,s_2, \dots, s_{|S|}\}$ be the set of candidate solution at the end of the while loop, then the algorithm filters out all the elements $s_t$ in $S$ such that $f(s_t \mid \{s_1, s_2, \dots, s_{t-1}\}) \le 0$. The resulting set is called $\overline S$. It is clear that $S$ and $\overline S$ respect the following two properties:
    \begin{itemize}
        \item $c(\bar{S})\le c(S)$
        \item $f(\bar{S}) \ge f(S)$
    \end{itemize}

    The analysis of the adaptive and query complexity of \thresholdgamma is similar to the one of \threshold, the only difference being the $\gamma$-ary search step. We have then the following result:
    \begin{lemma}
        \label{lem:non_monotone++_adaptivity}
            For any input $X \subseteq \N$, $\tau >0$, $\e \in (0,1)$, $B>0$, and $\gamma \in \{1,2,\dots,n\}$, \thresholdgamma terminates in $O\left(\frac{\log n \log (n\kappa_X)}{\e\log \gamma}\right)$ adaptive rounds and
        $O\left( \frac{n\gamma}{\e} \frac{\log n \log (n\kappa_X)}{\log \gamma}\right)$ value queries.
    \end{lemma}

    We then need to make sure that the good approximation properties of \threshold are retained. To this end, we prove two results, analogous to \Cref{lem:non_monotone_val,lem:non_monotone}. We start arguing that the expected marginal contribution of each element in $S$ is ``large enough''.
    \begin{lemma}
        \label{lem:non_monotone++_val}
       For any input $X \subseteq \N$, $\tau >0$, $\e \in (0,1)$, $B>0$, and $\gamma \in \{1,2,\dots,n\}$, the random sets $S$ computed by \thresholdgamma is such that $c(S) \le B$, and the following inequality holds true:
        \[
            {\E{f(S)}\ge \allowbreak (1-\e)^2\tau \,\E{c(S)}}.
        \]
    \end{lemma}
    \begin{proof}
        The proof of this Lemma is similar to that of \Cref{lem:non_monotone_val}, the only difference is that the new condition is stricter than the earlier one, so the inequality in \Cref{eq:value2} still holds.
    \end{proof}
    Note, the above Lemma, together with the properties of $S$ and $\overline S$ implies that 
    \begin{equation}
        \label{eq:aux_Sbar}
        \E{f(\overline S)}\ge \allowbreak (1-\e)^2\tau \,\E{c(S)}.
    \end{equation}

    Finally, we consider the set $G$ of the high value element with respect to $S$. 
    \begin{lemma}
    \label{lem:non_monotone++}
        For any input $X \subseteq \N$, $\tau >0$, $\e \in (0,1)$, $B>0$, and $\gamma \in \{1,2,\dots,n\}$, the random sets $S$ and $\overline S$ computed by \thresholdgamma are such that:
        \[f(\bar{S})\ge \e \ell \sum_{x \in G}f(x\,|\,S),
        \]
        where $G = \{x \in X \mid f(x \mid S) \ge \tau c(x)\}.$
    \end{lemma} 
\begin{proof}
    For all real numbers $a$ we denote with $a_+ = \max\{a,0\}$ its positive part and with $a_-= \max\{-a,0\}$ the negative one. Clearly $a = a_+ - a_-$.
    \begin{align*}
        f(S) &= \sum_{t=1}^T (f(s_t\,|\,S_{t})_+-f(s_t\,|\,S_{t})_-) \le \sum_{t=1}^T f(s_t\,|\,S_{t})_+ \le \sum_{s_t \in \bar{S}} f(s_t\,|\,\bar{S}_{t}) = f(\bar{S})\,,
    \end{align*}
    where in the last inequality we used submodularity. 

    Similarly to the proof for \threshold, consider $t_1,\ldots, t_{\ell}$, $E_{(1)},\allowbreak\dots,\allowbreak E_{(\ell)}$, $G_{(1)},\dots, G_{(\ell)}$, and $\calE_{(1)}, \allowbreak\dots,\allowbreak \calE_{(\ell)}$. Notice that they are all disjoint. Like before, for $s_i\in S$, $S_i$ denotes $\{s_1, \ldots, s_{i-1}\}$, but we slightly abuse the notation and have $S_{t_j}$ denote the set $S$ \emph{at the end} of the iteration of the outer while loop where $\textrm{ctr}$ is increased for the $\ell$th time, i.e., when the value condition is triggered for the $\ell^{th}$ time. We have
    \begin{align*}
        0 \le f\big(S_{t_\ell} \cup \textstyle\bigcup\limits_{j=1}^\ell E_{(j)}\big) &\le f(S_{t_\ell}) + f\big(\textstyle\bigcup\limits_{j=1}^\ell E_{(j)}\,|\,S_{t_\ell}\big) \\
        &\le \sum_{s_i \in S_{t_\ell}} (f(s_i\,|\,S_{i})_+-f(s_i\,|\,S_{i})_-) + \sum_{j=1}^\ell f(E_{(j)}\,|\,S_{t_j}) \\
        &\le \sum_{s_i \in S_{t_\ell}} (f(s_i\,|\,S_{i})_+-f(s_i\,|\,S_{i})_-) + \sum_{j=1}^\ell \sum_{x \in E_{(j)}} f(x\,|\,S_{t_j})\,.
    \end{align*}
    Rearranging terms, and using the value condition, we get
    \begin{align*}
        f(\bar{S}) \ge \sum_{s_i \in S_{t_\ell}} f(s_i\,|\,S_{i})_+ &\ge \sum_{s_i \in S_{t_\ell}} f(s_i\,|\,S_{i})_- +  \sum_{j=1}^\ell \sum_{x \in E_{(j)}} |f(x\,|\,S_{t_j})|\\
        &\ge \sum_{j = 1}^\ell\left[ \sum_{x \in E_{(j)}}|f(x\,|\,S_{t_j})| + \sum_{s_i \in \calE_{(j)}} |f(s_i\,|\,S_{i})|\right]\\
        &\ge\e\sum_{j = 1}^\ell \left[ \sum_{x \in G_{(j)}} f(x\,|\,S_{t_j})\right] \ge \ell \e \sum_{x \in G_{(\ell)}} f(x\,|\,S_{t_\ell})\,.
    \end{align*}
Observing that $S_{t_\ell} = S$ concludes the proof.
    \end{proof}
    
Using as routine \thresholdgamma instead of \threshold into \knapsack (let's call this modified algorithm \knapsackgamma) yields the main Theorem of the Section. Note that for $\gamma = n$, we recover the results of \Cref{thm:non_monotone}, while for $\gamma = 2$, i.e., binary search, we obtain a linear query complexity with $O(\log^2n)$ adaptive complexity. By the same analysis as in \Cref{thm:non_monotone} (with some minor modifications and using \Cref{lem:non_monotone++_val,lem:non_monotone++_adaptivity,lem:non_monotone++} instead than \Cref{lem:non_monotone_val,lem:non_monotone_adaptivity,lem:non_monotone}) we get the following result.
    
    \begin{theorem}
    \label{thm:non_monotone-binary}
        For any $\e \in (0,1)$ and $\gamma \in \{1,2,\dots,n\}$, \knapsackgamma with precision parameter $\e$ is a $(9.465 + O(\e))$-approximation algorithm that runs in $O(\tfrac{1}{\e}\frac{\log^2 n}{\log \gamma})$ adaptive rounds and issues $O(\frac{\gamma n\log^3 n\log{(\nicefrac 1\e)}}{\e^3\log \gamma})$ value queries.
    \end{theorem}
    
    If we plug in $\gamma = 2$ we get the first algorithm for non-monotone submodular maximization subject to knapsack constraint with sublinear adaptivity and nearly optimal $\tilde O(n)$ query complexity.
    \begin{corollary}
    \label{cor:non_monotone-binary}
        For any $\e \in (0,1)$, $2$-\knapsack with precision parameter $\e$ is a $(9.465 + O(\e))$-approximation algorithm that runs in $O(\tfrac{\log^2 n}{\e})$ adaptive rounds and issues $O(\frac{n}{\e^3}\log^3 n\log{\frac 1\e})$ value queries.
    \end{corollary}
    We conclude the section with a last remark: if we set $\gamma = n^{\delta}$, for some small (but positive) constant $\delta$, we obtain nearly optimal adaptive complexity of $O(\log n)$ and slightly-worst than-optimal query complexity of $O(n^{1+\delta})$.

\subsection{Monotone Submodular Functions}
\label{sec:monotone}

    For monotone objectives, the approximation ratio of \knapsack\ can be significantly improved, while the routine \threshold itself can be simplified. In particular, we do not need to address the value condition any more. Moreover, the small elements can be accounted for without any extra loss in the approximation. Beyond this, it is possible to trade a logarithmic loss in adaptivity for an almost linear gain in query complexity, as we did in the previous section. We start presenting the properties of \threshold when the underlying submodular function is monotone.

    \begin{lemma}
    \label{lem:monotone}
        Consider the problem of maximizing a monotone submodular function subject to a knapsack constraint. For any input $X \subseteq \N$,  $\tau>0$,  $\e \in (0,1)$,  and  $B>0$, the random set $S$ output by \threshold respects the following inequality:
        \[
        {\E{f(S)}\ge \tau(1-\e) \E{c(S)}}\,.
        \]
        Set $S$ is always a feasible solution and if $c(S) < B$, then all the elements in $X\setminus S$ have either marginal density with respect to $S$ smaller than $\tau$ or there is no room for them in the budget.
        Furthermore, the adaptivity is $O\left(\frac{1}{\varepsilon}\log \left(n\kappa_X\right)\right)$, while the query complexity is $O\left(\frac{n^2}{\varepsilon}\log \left(n\kappa_X\right)\right)$.
    \end{lemma}
\begin{proof}
    Again the proof is similar to the one for the non-monotone case in \Cref{thm:non_monotone}. There are three main differences. First, the adaptive complexity is given only by the number of times the cost condition is triggered, hence an upper bound is given by $\frac{1}{\varepsilon}\log \left(n\kappa_X\right)$. The query complexity is simply obtained multiplying that by a $n^2$ factor as in the proof of \Cref{thm:non_monotone}.
    
    Second, the algorithm can now only stop if the budget is exhausted or there are no good elements fitting within the budget; the while loop terminates only in those two cases.
    Finally, the main chain of inequalities is simply
    \[
        \mathbb{E}[f(s_t\,|\,S_{t})\,|\,\F_{t}] = \sum_{x \in X}p_x f(x\,|\,S_t) \ge \tau \sum_{x \in G}p_x c(x) \ge \tau (1-\e) \sum_{x \in X}p_x c(x) = (1-\e)\tau \,\E{c(s_t)\,|\,\F_t}\,,
    \]
    where in the second inequality we used the fact that the cost condition is not triggered. Taking the expectation on the whole process and reasoning as in the proof of \Cref{lem:non_monotone_val}, we conclude the proof. 
    \end{proof}

    We are ready to present the full algorithm for the monotone case \knapsackmonotone, whose details appear in the pseudocode. There are two main differences to the non-monotone case. First, there is no need to sample a subset $H$ and use the Sampling Lemma, since monotonicity implies that $f(S) \le f(S\cup O)$. Second, if one defines the small elements to be the ones with cost smaller than $\e \cdot \nicefrac Bn$ it is possible to account for them by simply adding all of them to the solution at the cost of filling an $\e$ fraction of the budget. Notice that this can be done while keeping $\kappa_{\N_+}$ linear in $n$. The remaining $(1-\e)$ fraction of the budget is then filled via \threshold on the large elements.

    \begin{algorithm}[!t]
        \caption{\knapsackmonotone}
        \label{alg:monotone_alg}
        \begin{spacing}{1.15}
        \begin{algorithmic}[1]
        \STATE \textbf{Environment:} Monotone Submodular function $f$, additive cost function $c$ and budget $B$
        \STATE \textbf{Input:} Precision parameter $\e \in (0,1)$ 
        \STATE $\alpha \gets \nicefrac 23$ 
        \STATE $\N_- \gets \{x \in \N: c(x) <  \nicefrac{\e B}{n}\}$, $\N_+ \leftarrow \N\setminus \mathcal{N_-}$ 
        \STATE $x^* \gets \argmax_{x \in \N}f(x)$, $\hat \tau \gets \alpha n \cdot \nicefrac {f(x^*)}{B}$ 
        \FOR{$i=0, 1,\dots,\lceil{\nicefrac {\log n}
{\e}\rceil}$ in parallel } 
        \STATE $\tau_i \gets \hat \tau \cdot (1-\e)^i$ 
        \FOR{$j=1, 2,\dots, \lceil\nicefrac 1\e \log(\nicefrac{1}{\e})\rceil$ in parallel} 
        \STATE $S^{i,j}\gets \threshold(\N_+,\tau_i,\e,(1-\e)B)$ 
        \STATE $T^{i,j} \gets S^{i,j} \cup \N_-$
        \ENDFOR
        \ENDFOR
        \STATE  $T\gets \argmax_{i,j}\{f(T^{i,j}),f(x^*)\}$
    \STATE \textbf{Return} $T$
    \end{algorithmic}
    \end{spacing}
    \end{algorithm}

\begin{theorem}
\label{thm:monotone}
    Consider monotone submodular maximization subject to a knapsack constraint. For any $\e \in (0,1)$, \knapsackmonotone with precision parameter $\e$ achieves a $(3+O(\e))$-approximation algorithm that runs in $O(\frac{1}{\e}\log n)$ adaptive rounds and  issues $O(\frac{n^2}{\e^3}\log^2 n\log\frac{1}{\e})$ value queries. 
\end{theorem}

\begin{proof}
    The bounds on adaptive and query complexity descend by combining \Cref{lem:monotone}, the fact that the thresholds are guessed in parallel, and the fact that $\kappa_{\N_{+}} \in O(\nicefrac{n}\e)$. 

    We move our attention to the approximation guarantee. Let $O^*$ be the optimal solution and let $\tau^*=\alpha \frac{f(O^*)}{B}$. By the usual argument we have that there exists a $\tau_{\hat i}$ such that ${(1-\hat \e)\tau^* \le \tau_{\hat i} < \tau^*}$. We focus on what happens for that $\tau_{\hat i}$ (so that we may omit the dependence on the index ${\hat i}$). Denote with $S$ the generic random output of \threshold with input $\N_+, \tau_{\hat i}, \e$ and $(1-\e)B$. We have two cases, depending on the expected cost of $S$. 
    If $\E{c(S)} \ge \frac B2(1- \e)^2$, then we apply \Cref{lem:monotone} and for any repetition $j$ of the inner for loop we have 
    \begin{equation}
    \label{eq:monotone_budget}
        f(O^*) \le \frac{2}{\alpha(1-\e)^4}f(S^{\hat i,j}) \le (3 + O(\e))f(T)\,.
    \end{equation}

Let's now address the other case, i.e., $\E{c(S)} < \frac B2(1 - \e)^2$. By standard probabilistic arguments, with probability at least $(1-\e)$, at least one of the iterations of the inner for loop (corresponding to $\tau_{\hat i}$) it holds that the realized $c(S^{{\hat i},j})$ is indeed smaller than $\frac B2(1 - \e)$. We call $\G$ the event that such property holds in at least one iteration (call $\hat j$ such iteration). Clearly there may be at most one item which belongs to the optimal solution $O^*$ but not in $S^{\hat i,\hat j}$. 
    \begin{align*}
        f(O^*) &\le f(S \cup O^*)\\
        &\le f(T \cup (O^*\setminus \N_-)) \le f(T) + f(\tilde x\,|\,T) + \sum_{x \in O^*\setminus \{\tilde x\}} f(x\,|\,T) \\
        &\le f(T) + f(x^*)  +  \sum_{x \in O^*\setminus \{\tilde x\}} f(x\,|\,S) \le f(T) + (f(x^*) - \alpha \tfrac{f(O^*)}2)+ \alpha f(O^*)\,.
    \end{align*}
    If $\nicefrac 2\alpha f(x^*) \ge \OPT$, then for sure our algorithm is a $3$ approximation, so we assume the contrary, which entails that the term $(f(x^*) - \alpha \tfrac{f(O^*)}2)$ in the previous inequality is negative and can be ignored.
    All in all, we have proved that under $\G$ the following inequality holds:
    \begin{equation*}
        f(O^*) \le (1-\alpha) f(S^{\hat i,\hat j}) \le 3 f(T).
    \end{equation*}
    By taking the expected value and using the fact that $\G$ has probability $1-\e$ we get that 
    \begin{equation}
    \label{eq:monotone_G}
        f(O^*) \le \frac{3}{1-\e} f(T) = (3 + O(\e)) f(T).
    \end{equation}
    The approximation result follows by combining together \eqref{eq:monotone_budget} and \eqref{eq:monotone_G}.
    \end{proof}

    When the submodular function is monotone, in  \threshold the only meaningful condition to address is the cost one, as the value condition is always respected. This implies that the task of finding $k^*$ in \threshold can be directly performed by $\gamma$-ary search, without any modification. If we call \knapsackgammamonotone the version of \knapsackmonotone that uses $\gamma$-ary search inside \threshold, we get the following result.
    \begin{corollary}
    \label{cor:monotone}
        Consider the problem of monotone submodular maximization subject to a knapsack constraint. For any $\e \in (0,1)$ and any $\gamma \in \{1,2,\dots,n\}$, \knapsackgammamonotone with precision parameter $\e$ achieves a $(3+O(\e))$-approximation algorithm that runs in $O(\frac{\log^2 n}{\e\log \gamma })$ adaptive rounds and $O(\frac{\gamma n}{\log \gamma \e^3}\log^3 n\log{\frac 1\e})$ value queries. 
    \end{corollary}
Note that the variant using $\tilde{O}(n)$ queries is the first $O(1)$-approximation algorithm for the problem combining this few queries with sublinear adaptivity.



\subsection{Cardinality constraints}
\label{sec:cardinality}
    Cardinality constraints is a special case of knapsack constraint, when the cost function is simply the counting function $c(S) = |S|$. Therefore \knapsack\ can be directly applied to cardinality constraints for (possibly) non-monotone objectives. Again, with some simple modifications, it is possible to achieve a much better approximation. In presence of cardinality constraints, in fact, there is no need to address separately small and large elements, as all elements impact in the same way on the constraint. See the pseudocode of \cardinal for further details. Moreover, if \threshold stops without filling the cardinality $k$, it means that no other element clears the threshold on the marginal contribution.

    \begin{algorithm}
    \caption{\cardinal}
    \label{alg:non_monotone_alg_card}
        \begin{spacing}{1.1}
        \begin{algorithmic}[1]
        \STATE \textbf{Environment:} Submodular function $f$ and cardinality constraint $k$
        \STATE \textbf{Input:} Precision parameter $\e \in (0,1)$ 
        \STATE $p \gets \sqrt{2}-1$,  $\alpha \gets 3 - 2\sqrt{2} $
        \STATE $x^* \gets \argmax_{x \in \N}f(x)$, $\hat \tau \gets \alpha n \cdot \nicefrac {f(x^*)}{B}$
        \FOR{$i=0, 1,\dots,\lceil{\nicefrac {\log n}
{\e}\rceil}$ in parallel } 
        \STATE $\tau_i \gets \hat \tau \cdot (1-\e)^i$ 
        \FOR{$j=1, 2,\dots, \lceil\nicefrac 1\e \log(\nicefrac{1}{\e})\rceil$ in parallel} 
        \STATE Sample independently a subset $H^{i,j} \gets \N(p)$
        \STATE $S^{i,j}\gets \threshold(H^{i,j},\tau_i,\e,k)$
        \ENDFOR
        \STATE Sample independently a subset $H_i \gets \N(p)$
        \FOR{$\ell = 1, 2, \dots, \lceil\nicefrac 1\e \log(\nicefrac{1}{\e})\rceil$ in parallel} 
            \STATE $S_{i,\ell}\gets \threshold(H_i,\tau_i,\e,k)$
        \ENDFOR
        \ENDFOR
        \STATE $S_{\ge} \in \argmax_{i,j} \{f(S^{i,j})\}$, $S_{<} \in \argmax_{i,\ell} \{f(S_{i,\ell})\}$ 
        \STATE \textbf{return}  $T \in\argmax\{f(S_{\ge}),f(S_<),f(x^*)\}$
        \end{algorithmic}
        \end{spacing}
        \end{algorithm}

\begin{theorem}
\label{thm:cardinality}
    Consider the problem of non-monotone submodular maximization subject to a cardinality constraint $k$. For any $\e \in (0,1)$, \cardinal with precision parameter $\e$ is a $(5.83 + O(\e))$-approximation algorithm that runs in $O(\frac1\e \log n)$ adaptive rounds and issues $O(\frac{nk}{\e^3}\log n \log k \log \frac{1}{\e})$ value queries.
\end{theorem}
\begin{proof}
    The adaptive and query complexity are as in the knapsack case (\Cref{thm:non_monotone}); the only difference being that now each sequence drawn from \sampling \ has at most length $k$.

    For the approximation guarantees our argument closely follows the same step as the ones in \Cref{thm:non_monotone}, so we only report the differences. The crucial difference is that now we say that $H$ fills the budget is, in expectation, the cardinality of the output of \threshold on $H$ is at least $(1-\e)^2k.$
    With this difference it is easy to see that the cardinality version of \Cref{cl:bound_ge} yields the following inequality:
    \[
                    \OPT \le \frac{1}{\alpha(1-\e)^6} \E{f(S_\ge)} \le (3 + 2\sqrt{2} + O(\e))\ALG.
                \]
    Consider now the other case, i.e., when the random $H$ does not fill the budget. The same arguments as in \Cref{cl:p_bound,cl:anyreal,cl:bound_le} carry over, with the simplification that if the budget (cardinality $k$) is not filled, then no element in $H$ but not selected by \threshold has marginal density larger than $\tau$. In particular, by \Cref{eq:very_last} and plugging the current value of $p$ and $\alpha$ 
    we get the desired result. Note, with respect to \Cref{thm:non_monotone} we do not need to account for the ``small'' elements.
\end{proof}

    Also for this problem we can use \thresholdgamma instead of \threshold, this gives the algorithm \cardinalgamma whose properties are summarized in the following Corollary.

    \begin{corollary}
    \label{cor:cardinality}
          Consider the problem of non-monotone submodular maximization subject to a cardinality constraint $k$. For any $\e \in (0,1)$ and $\gamma \in \{1,2,\dots,k\}$, \cardinalgamma with precision parameter $\e$ is a $(5.83 + \e)$-approximation algorithm that runs in $O(\frac{\log k}{\e\log \gamma }\log n)$ adaptive rounds and issues $O(\frac{\gamma n}{\e^3\log \gamma }\log n\log^2 k\log{\frac 1\e})$ value queries.
    \end{corollary}

    We highlight that, although we do not heavily adjust our algorithms to cardinality constraints, \Cref{thm:cardinality} is directly comparable to the results of \citet{EneN20} and \citet{Kuhnle21} which are tailored for this specific problem.

\section{Conclusions}

In this paper we close the gap for the adaptive complexity of  non-monotone submodular maximization subject to a knapsack constraint, up to a $O(\log \log n)$ factor. Our algorithm, \knapsack, is combinatorial and can be modified to achieve  trade-offs between adaptivity and query complexity. In particular, it may use nearly linear queries, while achieving an exponential improvement on adaptivity compared to existing algorithms with subquadratic query complexity. 

\section*{Acknowledgement}

This work was supported by the ERC Advanced Grant 788893 AMDROMA “Algorithmic and Mechanism Design Research in Online Markets”, the MIUR PRIN project ALGADIMAR “Algorithms, Games, and Digital Markets”, and the NWO Veni project No. VI.Veni.192.153.

\bibliographystyle{plainnat}
\bibliography{references}

\begin{thebibliography}{47}
\providecommand{\natexlab}[1]{#1}
\providecommand{\url}[1]{\texttt{#1}}
\expandafter\ifx\csname urlstyle\endcsname\relax
  \providecommand{\doi}[1]{doi: #1}\else
  \providecommand{\doi}{doi: \begingroup \urlstyle{rm}\Url}\fi

\bibitem[Amanatidis et~al.(2021)Amanatidis, Fusco, Lazos, Leonardi,
  Marchetti{-}Spaccamela, and Reiffenh{\"{a}}user]{AmanatidisFLLMR21}
Georgios Amanatidis, Federico Fusco, Philip Lazos, Stefano Leonardi, Alberto
  Marchetti{-}Spaccamela, and Rebecca Reiffenh{\"{a}}user.
\newblock Submodular maximization subject to a knapsack constraint:
  Combinatorial algorithms with near-optimal adaptive complexity.
\newblock In \emph{{ICML}}, volume 139 of \emph{Proceedings of Machine Learning
  Research}, pages 231--242. {PMLR}, 2021.

\bibitem[Amanatidis et~al.(2022)Amanatidis, Fusco, Lazos, Leonardi, and
  Reiffenh{\"{a}}user]{AmanatidisFLLR22}
Georgios Amanatidis, Federico Fusco, Philip Lazos, Stefano Leonardi, and
  Rebecca Reiffenh{\"{a}}user.
\newblock Fast adaptive non-monotone submodular maximization subject to a
  knapsack constraint.
\newblock \emph{J. Artif. Intell. Res.}, 74:\penalty0 661--690, 2022.

\bibitem[Balkanski and Singer(2018)]{BalkanskiS18}
Eric Balkanski and Yaron Singer.
\newblock The adaptive complexity of maximizing a submodular function.
\newblock In \emph{{STOC}}, pages 1138--1151. {ACM}, 2018.

\bibitem[Balkanski et~al.(2018)Balkanski, Breuer, and Singer]{BalkanskiBS18}
Eric Balkanski, Adam Breuer, and Yaron Singer.
\newblock Non-monotone submodular maximization in exponentially fewer
  iterations.
\newblock In \emph{NeurIPS}, pages 2359--2370, 2018.

\bibitem[Balkanski et~al.(2019)Balkanski, Rubinstein, and
  Singer]{BalkanskiRS19}
Eric Balkanski, Aviad Rubinstein, and Yaron Singer.
\newblock An exponential speedup in parallel running time for submodular
  maximization without loss in approximation.
\newblock In \emph{{SODA}}, pages 283--302. {SIAM}, 2019.

\bibitem[Balkanski et~al.(2022)Balkanski, Rubinstein, and
  Singer]{BalkanskiRS19matroid}
Eric Balkanski, Aviad Rubinstein, and Yaron Singer.
\newblock An optimal approximation for submodular maximization under a matroid
  constraint in the adaptive complexity model.
\newblock \emph{Oper. Res.}, 70\penalty0 (5):\penalty0 2967--2981, 2022.

\bibitem[Breuer et~al.(2020)Breuer, Balkanski, and Singer]{BreuerBS20}
Adam Breuer, Eric Balkanski, and Yaron Singer.
\newblock The {FAST} algorithm for submodular maximization.
\newblock In \emph{{ICML}}, volume 119 of \emph{Proceedings of Machine Learning
  Research}, pages 1134--1143. {PMLR}, 2020.

\bibitem[Buchbinder and Feldman(2018)]{BuchbinderF18}
Niv Buchbinder and Moran Feldman.
\newblock Submodular functions maximization problems.
\newblock In \emph{Handbook of Approximation Algorithms and Metaheuristics
  {(1)}}, pages 753--788. Chapman and Hall/CRC, 2018.

\bibitem[Buchbinder et~al.(2014)Buchbinder, Feldman, Naor, and
  Schwartz]{BuchbinderFNS14}
Niv Buchbinder, Moran Feldman, Joseph Naor, and Roy Schwartz.
\newblock Submodular maximization with cardinality constraints.
\newblock In \emph{{SODA}}, pages 1433--1452. {SIAM}, 2014.

\bibitem[Chekuri and Quanrud(2019{\natexlab{a}})]{ChekuriQ19}
Chandra Chekuri and Kent Quanrud.
\newblock Submodular function maximization in parallel via the multilinear
  relaxation.
\newblock In \emph{{SODA}}, pages 303--322. {SIAM}, 2019{\natexlab{a}}.

\bibitem[Chekuri and Quanrud(2019{\natexlab{b}})]{ChekuriQ19matroid}
Chandra Chekuri and Kent Quanrud.
\newblock Parallelizing greedy for submodular set function maximization in
  matroids and beyond.
\newblock In \emph{{STOC}}, pages 78--89. {ACM}, 2019{\natexlab{b}}.

\bibitem[Chekuri et~al.(2014)Chekuri, Vondr{\'{a}}k, and
  Zenklusen]{ChekuriVZ14}
Chandra Chekuri, Jan Vondr{\'{a}}k, and Rico Zenklusen.
\newblock Submodular function maximization via the multilinear relaxation and
  contention resolution schemes.
\newblock \emph{{SIAM} J. Comput.}, 43\penalty0 (6):\penalty0 1831--1879, 2014.

\bibitem[Chen et~al.(2019)Chen, Feldman, and Karbasi]{FK19}
Lin Chen, Moran Feldman, and Amin Karbasi.
\newblock Unconstrained submodular maximization with constant adaptive
  complexity.
\newblock In \emph{{STOC}}, pages 102--113. {ACM}, 2019.

\bibitem[Chen et~al.(2021)Chen, Dey, and Kuhnle]{ChenDK21}
Yixin Chen, Tonmoy Dey, and Alan Kuhnle.
\newblock Best of both worlds: Practical and theoretically optimal submodular
  maximization in parallel.
\newblock In \emph{NeurIPS}, pages 25528--25539, 2021.

\bibitem[Cui et~al.(2023{\natexlab{a}})Cui, Han, Tang, Huang, Li, and
  Zhiyuli]{Cui000LZ23}
Shuang Cui, Kai Han, Jing Tang, He~Huang, Xueying Li, and Aakas Zhiyuli.
\newblock Practical parallel algorithms for submodular maximization subject to
  a knapsack constraint with nearly optimal adaptivity.
\newblock In \emph{{AAAI}}, pages 7261--7269. {AAAI} Press, 2023{\natexlab{a}}.

\bibitem[Cui et~al.(2023{\natexlab{b}})Cui, Han, Tang, Huang, Li, Zhiyuli, and
  Li]{CuiH23}
Shuang Cui, Kai Han, Jing Tang, He~Huang, Xueying Li, Aakas Zhiyuli, and
  Hanxiao Li.
\newblock Practical parallel algorithms for non-monotone submodular
  maximization.
\newblock \emph{CoRR}, abs/2308.10656, 2023{\natexlab{b}}.

\bibitem[Das and Kempe(2008)]{DasK08}
Abhimanyu Das and David Kempe.
\newblock Algorithms for subset selection in linear regression.
\newblock In \emph{{STOC}}, pages 45--54. {ACM}, 2008.

\bibitem[Das and Kempe(2018)]{DasK18}
Abhimanyu Das and David Kempe.
\newblock Approximate submodularity and its applications: Subset selection,
  sparse approximation and dictionary selection.
\newblock \emph{J. Mach. Learn. Res.}, 19:\penalty0 3:1--3:34, 2018.

\bibitem[Dueck and Frey(2007)]{DueckF07}
Delbert Dueck and Brendan~J. Frey.
\newblock Non-metric affinity propagation for unsupervised image
  categorization.
\newblock In \emph{{ICCV}}, pages 1--8. {IEEE} Computer Society, 2007.

\bibitem[Duetting et~al.(2022)Duetting, Fusco, Lattanzi, Norouzi{-}Fard, and
  Zadimoghaddam]{DuettingFLNZ22}
Paul Duetting, Federico Fusco, Silvio Lattanzi, Ashkan Norouzi{-}Fard, and
  Morteza Zadimoghaddam.
\newblock Deletion robust submodular maximization over matroids.
\newblock In \emph{{ICML}}, volume 162 of \emph{Proceedings of Machine Learning
  Research}, pages 5671--5693. {PMLR}, 2022.

\bibitem[Ene and Nguyen(2019)]{EneN19}
Alina Ene and Huy~L. Nguyen.
\newblock Submodular maximization with nearly-optimal approximation and
  adaptivity in nearly-linear time.
\newblock In \emph{{SODA}}, pages 274--282. {SIAM}, 2019.

\bibitem[Ene and Nguyen(2020)]{EneN20}
Alina Ene and Huy~L. Nguyen.
\newblock Parallel algorithm for non-monotone dr-submodular maximization.
\newblock In \emph{{ICML}}, volume 119 of \emph{Proceedings of Machine Learning
  Research}, pages 2902--2911. {PMLR}, 2020.

\bibitem[Ene et~al.(2018)Ene, Nguyen, and Vladu]{EneNV18}
Alina Ene, Huy~L. Nguyen, and Adrian Vladu.
\newblock A parallel double greedy algorithm for submodular maximization.
\newblock \emph{CoRR}, abs/1812.01591, 2018.

\bibitem[Ene et~al.(2019)Ene, Nguyen, and Vladu]{EneNV19}
Alina Ene, Huy~L. Nguyen, and Adrian Vladu.
\newblock Submodular maximization with matroid and packing constraints in
  parallel.
\newblock In \emph{{STOC}}, pages 90--101. {ACM}, 2019.

\bibitem[Esfandiari et~al.(2021)Esfandiari, Karbasi, and
  Mirrokni]{EsfandiariKM21}
Hossein Esfandiari, Amin Karbasi, and Vahab~S. Mirrokni.
\newblock Adaptivity in adaptive submodularity.
\newblock In \emph{{COLT}}, volume 134 of \emph{Proceedings of Machine Learning
  Research}, pages 1823--1846. {PMLR}, 2021.

\bibitem[Fahrbach et~al.(2019{\natexlab{a}})Fahrbach, Mirrokni, and
  Zadimoghaddam]{FahrbachMZ19}
Matthew Fahrbach, Vahab~S. Mirrokni, and Morteza Zadimoghaddam.
\newblock Submodular maximization with nearly optimal approximation, adaptivity
  and query complexity.
\newblock In \emph{{SODA}}, pages 255--273. {SIAM}, 2019{\natexlab{a}}.

\bibitem[Fahrbach et~al.(2019{\natexlab{b}})Fahrbach, Mirrokni, and
  Zadimoghaddam]{FahrbachMZ19nonmonotone}
Matthew Fahrbach, Vahab~S. Mirrokni, and Morteza Zadimoghaddam.
\newblock Non-monotone submodular maximization with nearly optimal adaptivity
  and query complexity.
\newblock In \emph{{ICML}}, volume~97 of \emph{Proceedings of Machine Learning
  Research}, pages 1833--1842. {PMLR}, 2019{\natexlab{b}}.

\bibitem[Feige(1998)]{Feige98}
Uriel Feige.
\newblock A threshold of ln \emph{n} for approximating set cover.
\newblock \emph{J. {ACM}}, 45\penalty0 (4):\penalty0 634--652, 1998.

\bibitem[Feige et~al.(2011)Feige, Mirrokni, and Vondr{\'{a}}k]{FeigeMV11}
Uriel Feige, Vahab~S. Mirrokni, and Jan Vondr{\'{a}}k.
\newblock Maximizing non-monotone submodular functions.
\newblock \emph{{SIAM} J. Comput.}, 40\penalty0 (4):\penalty0 1133--1153, 2011.

\bibitem[Feldman et~al.(2011)Feldman, Naor, and Schwartz]{FeldmanNS11}
Moran Feldman, Joseph Naor, and Roy Schwartz.
\newblock A unified continuous greedy algorithm for submodular maximization.
\newblock In \emph{{FOCS}}, pages 570--579. {IEEE} Computer Society, 2011.

\bibitem[Gupta et~al.(2010)Gupta, Roth, Schoenebeck, and Talwar]{GuptaRST10}
Anupam Gupta, Aaron Roth, Grant Schoenebeck, and Kunal Talwar.
\newblock Constrained non-monotone submodular maximization: Offline and
  secretary algorithms.
\newblock In \emph{{WINE}}, volume 6484 of \emph{Lecture Notes in Computer
  Science}, pages 246--257. Springer, 2010.

\bibitem[Hartline et~al.(2008)Hartline, Mirrokni, and
  Sundararajan]{HartlineMS08}
Jason~D. Hartline, Vahab~S. Mirrokni, and Mukund Sundararajan.
\newblock Optimal marketing strategies over social networks.
\newblock In \emph{{WWW}}, pages 189--198. {ACM}, 2008.

\bibitem[Kazemi et~al.(2018)Kazemi, Zadimoghaddam, and Karbasi]{KazemiZK18}
Ehsan Kazemi, Morteza Zadimoghaddam, and Amin Karbasi.
\newblock Scalable deletion-robust submodular maximization: Data summarization
  with privacy and fairness constraints.
\newblock In \emph{{ICML}}, volume~80 of \emph{Proceedings of Machine Learning
  Research}, pages 2549--2558. {PMLR}, 2018.

\bibitem[Kempe et~al.(2015)Kempe, Kleinberg, and Tardos]{KempeKT15}
David Kempe, Jon~M. Kleinberg, and {\'{E}}va Tardos.
\newblock Maximizing the spread of influence through a social network.
\newblock \emph{Theory Comput.}, 11:\penalty0 105--147, 2015.

\bibitem[Khanna et~al.(2017)Khanna, Elenberg, Dimakis, Negahban, and
  Ghosh]{KhannaEDNG17}
Rajiv Khanna, Ethan~R. Elenberg, Alexandros~G. Dimakis, Sahand~N. Negahban, and
  Joydeep Ghosh.
\newblock Scalable greedy feature selection via weak submodularity.
\newblock In \emph{{AISTATS}}, volume~54 of \emph{Proceedings of Machine
  Learning Research}, pages 1560--1568. {PMLR}, 2017.

\bibitem[Kuhnle(2021)]{Kuhnle21}
Alan Kuhnle.
\newblock Nearly linear-time, parallelizable algorithms for non-monotone
  submodular maximization.
\newblock In \emph{{AAAI}}, pages 8200--8208. {AAAI} Press, 2021.

\bibitem[Kulik et~al.(2013)Kulik, Shachnai, and Tamir]{KulikST13}
Ariel Kulik, Hadas Shachnai, and Tami Tamir.
\newblock Approximations for monotone and nonmonotone submodular maximization
  with knapsack constraints.
\newblock \emph{Math. Oper. Res.}, 38\penalty0 (4):\penalty0 729--739, 2013.

\bibitem[Lattanzi et~al.(2020)Lattanzi, Mitrovic, Norouzi{-}Fard, Tarnawski,
  and Zadimoghaddam]{LattanziMNTZ20}
Silvio Lattanzi, Slobodan Mitrovic, Ashkan Norouzi{-}Fard, Jakub Tarnawski, and
  Morteza Zadimoghaddam.
\newblock Fully dynamic algorithm for constrained submodular optimization.
\newblock In \emph{NeurIPS}, 2020.

\bibitem[Minoux(1978)]{Minoux78}
Michel Minoux.
\newblock Accelerated greedy algorithms for maximizing submodular set
  functions.
\newblock In \emph{Optimization Techniques}, pages 234--243, Berlin,
  Heidelberg, 1978. Springer Berlin Heidelberg.

\bibitem[Mirzasoleiman et~al.(2013)Mirzasoleiman, Karbasi, Sarkar, and
  Krause]{MirzasoleimanKSK13}
Baharan Mirzasoleiman, Amin Karbasi, Rik Sarkar, and Andreas Krause.
\newblock Distributed submodular maximization: Identifying representative
  elements in massive data.
\newblock In \emph{{NIPS}}, pages 2049--2057, 2013.

\bibitem[Mirzasoleiman et~al.(2016)Mirzasoleiman, Badanidiyuru, and
  Karbasi]{MirzasoleimanBK16}
Baharan Mirzasoleiman, Ashwinkumar Badanidiyuru, and Amin Karbasi.
\newblock Fast constrained submodular maximization: Personalized data
  summarization.
\newblock In \emph{{ICML}}, volume~48 of \emph{{JMLR} Workshop and Conference
  Proceedings}, pages 1358--1367. JMLR.org, 2016.

\bibitem[Mirzasoleiman et~al.(2020)Mirzasoleiman, Bilmes, and
  Leskovec]{MirzasoleimanBL20}
Baharan Mirzasoleiman, Jeff~A. Bilmes, and Jure Leskovec.
\newblock Coresets for data-efficient training of machine learning models.
\newblock In \emph{{ICML}}, volume 119 of \emph{Proceedings of Machine Learning
  Research}, pages 6950--6960. {PMLR}, 2020.

\bibitem[Nemhauser and Wolsey(1978)]{NemhauserW78}
George~L. Nemhauser and Laurence~A. Wolsey.
\newblock Best algorithms for approximating the maximum of a submodular set
  function.
\newblock \emph{Math. Oper. Res.}, 3\penalty0 (3):\penalty0 177--188, 1978.

\bibitem[Nemhauser et~al.(1978)Nemhauser, Wolsey, and Fisher]{NemhauserWF78}
George~L. Nemhauser, Laurence~A. Wolsey, and Marshall~L. Fisher.
\newblock An analysis of approximations for maximizing submodular set functions
  - {I}.
\newblock \emph{Math. Program.}, 14\penalty0 (1):\penalty0 265--294, 1978.

\bibitem[Schrijver(2003)]{Schrijver03}
Alexander Schrijver.
\newblock \emph{Combinatorial optimization: polyhedra and efficiency},
  volume~24.
\newblock Springer, 2003.

\bibitem[Sviridenko(2004)]{Sviridenko04}
Maxim Sviridenko.
\newblock A note on maximizing a submodular set function subject to a knapsack
  constraint.
\newblock \emph{Oper. Res. Lett.}, 32\penalty0 (1):\penalty0 41--43, 2004.

\bibitem[Tschiatschek et~al.(2014)Tschiatschek, Iyer, Wei, and
  Bilmes]{TschiatschekIWB14}
Sebastian Tschiatschek, Rishabh~K. Iyer, Haochen Wei, and Jeff~A. Bilmes.
\newblock Learning mixtures of submodular functions for image collection
  summarization.
\newblock In \emph{{NIPS}}, pages 1413--1421, 2014.

\end{thebibliography}

\appendix

\section{Additional Combinatorial Properties}
\label{app:combinatorial}

    In this Section, we recall the definitions of some of the combinatorial notions we used through the paper.
    
    A set system $(\N,\mathcal I)$ is a family $\mathcal I \subseteq 2^{\N}$ of feasible subset of a ground set $\N$. Given a set system $(\N,\mathcal I)$ and a set $A\subseteq \N$, a base of $A$ is any feasible set $I \subseteq A$, $I \in \mathcal I$ that is maximal with respect to inclusion: $I \cup \{a\} \notin \mathcal I$ for any $a \in A \setminus I$. A natural complexity measure of a set system is given by the ratio of the largest and smallest basis that is possible to recover in any subset of the ground set.
    \begin{definition}[$k$-systems.]
        A set system $(\N,\mathcal I)$ is said to be a $k$-system if the following inequality holds true for any $A \subseteq \N$:
        \[
            \frac{\max_{B \in \mathcal B(A)} |B|}{\min_{B \in \mathcal B(A)} |B|} \le k,
        \]
        where $\mathcal B(A)$ is the family of all the basis contained in $A$. We adopt the convention that $\max$ over the empty set is $0$ and $\min$ over the empty set is $+\infty$.
    \end{definition}

    With this definition we can immediately argue that the set system generated by a knapsack constraint may be a $\Theta(n)$-system. This tells us that there may be an extremely large gap in cardinality between two different maximal solutions of the submodular maximization problem subject to a knapsack constraint. 

    \begin{example}[Knapsack are $\Theta(n)$-systems]
        Let $\N = \{x_1, x_2, \dots, x_n\}$ be a set of $n$ elements and $c$ be an additive cost function on $\N$ such that $c(x_i) = 1$ for all $i < n$, and $c(x_n) = n$. Consider now the set system $(\N,\mathcal I)$ induced by the knapsack constraint of budget $n$, i.e., all the subsets of $\N$ whose cost is at most $n$. There are two basis for this set system: $\{x_1, \dots, x_{n-1}\}$ and $\{x_n\}$, thus $(\N,\mathcal I)$ is an $(n-1)$-system.
    \end{example}

    While knapsack constraints are characterized by this wildly varying cardinality between basis, there is another well studied class of constraints that are $1$-systems.

    \begin{definition}[Matroids.]
    A non-empty set system $(\N,\mathcal I)$, with $\mathcal I \subseteq 2^{\N}$ is called a \emph{matroid} if it satisfies the following properties:
    \begin{itemize}
        \item[$(i)$]  \textit{Downward-closure}: if $A \subseteq B$ and $B \in \mathcal I$, then $A \in \mathcal I$; 
        \item[$(ii)$] \textit{Augmentation}: if $A, B \in \mathcal I$ with $|A| < |B|$, then there exists $e \in B$ such that $A + e \in \mathcal I$.
    \end{itemize} 
    \end{definition}
    \noindent By the augmentation property, is immediate to verify that matroids are $1$-systems.
\end{document}